\documentclass[a4paper,english]{article}
\usepackage[utf8]{inputenc}
\usepackage[stmaryrd,textwidth=2.5cm,todo=false,tikz=true]{generic}
\usepackage[margin=3.25cm]{geometry}
\usepackage{logic9}
\definecolor{nicecyan}{HTML}{006165}
\definecolor{nicered}{HTML}{DB3A34}

\usetikzlibrary{patterns,shadows,shadings,shapes,arrows,calc,backgrounds}
\usetikzlibrary{decorations.markings,decorations.pathmorphing,decorations.pathreplacing}
\tikzstyle{dot} = [draw,shape=circle,fill, minimum size=1mm, inner sep=0pt, outer sep=0pt]
\tikzstyle{ball} = [draw,shape=circle, minimum size=1cm, inner sep=0pt, outer sep=0pt]
\tikzstyle{arrow} = [->, >=stealth, line width=0.5pt, shorten >=1pt, shorten <=1pt, 
                     preaction={-, draw, shorten >=1pt, shorten <=1pt, line width=1.75pt, white}]
\tikzstyle{link} = [line width=0.5pt, shorten >=1pt, shorten <=1pt, 
                    preaction={-, draw, shorten >=1pt, shorten <=1pt, line width=1.75pt, white}]

\NewDocumentCommand{\T}{O{T}}{{#1}}
\NewDocumentCommand{\Tfull}{O{T} D<>{Q} D<>{U} D<>{I} D<>{E} D<>{F} D<>{x_\out}}
                   {#1 = (\Sigma,\Gamma,X,#2,#3,#4,#5,#6,#7)}
\NewDocumentCommand{\Tfulln}{O{T} D<>{Q} D<>{U} D<>{I} D<>{E} D<>{F} D<>{x_1}}
                   {#1 = (\Sigma,\Gamma,X,#2,#3,#4,#5,#6,#7)}

\newcommand{\dom}[1]{{\text{Dom}(#1)}}
\newcommand{\rng}[1]{{\text{Rng}(#1)}}

\newcommand{\out}{{\text{out}}}
\newcommand{\val}[1]{\text{val}_{#1}}
\newcommand{\pos}{\mathbin{\bullet}}
\newcommand{\Val}[1]{\text{Val}_{#1}}
\newcommand{\gap}[1]{\langle #1 \rangle}
\newcommand{\dlen}[1]{\lVert#1\rVert}
\renewcommand{\mid}{\::\:}
\newcommand{\dual}[1]{{#1}^\star}
\newcommand{\closure}[2][\alpha]{{#2}^{\uparrow #1}}
\newcommand{\App}{A}
\newcommand{\Appbis}{B}
\newcommand{\eff}[1]{\text{Effect}(#1)}
\NewDocumentCommand{\pump}{D<>{\bar\g} O{\bar n} m}{\text{pump}_{#1}^{#2}(#3)}
\newcommand{\yields}[1][]{\mathrel{\unrhd_{#1}}}

\newcommand{\normflows}{2^m}
\newcommand{\edgeambiguity}{k\cdot2^m}
\newcommand{\bound}[1][\normflows]{m\,c\,|Q|\,2^{3\cdot#1}}
\newcommand{\boundnom}[1][\normflows]{c\,|Q|\,2^{3\cdot#1}}
\newcommand{\boundnomc}[1][\normflows]{|Q|\,2^{3\cdot#1}}
\newcommand{\boundnomcQ}[1][\normflows]{2^{3\cdot#1}}

\title{\bf Equivalence of finite-valued streaming string transducers is decidable}

\author{Anca Muscholl\\ LaBRI, University of Bordeaux \and
        Gabriele Puppis\\ CNRS, LaBRI}
\date{}

\begin{document}

\maketitle

\begin{abstract}
In this paper we provide a positive answer to a question left open by 
Alur and and Deshmukh in 2011 by showing that equivalence of finite-valued
copyless streaming string transducers is decidable. \end{abstract}

\section{Introduction}

Finite transducers are simple devices that allow to reason about data
transformations in an effective, and even efficient way. In their most
basic form they transform strings using finite
control. Unlike automata, their power heavily depends on various
parameters, like non-determinism, the capability of scanning the input
several times, or the kind of storage they may use. The oldest
transducer model, known as generalized sequential machine, 
extends finite automata by outputs. Inspired by an
approach that applies to arbitrary relational structures~\cite{CE12}, 
logic-based transformations (also called transductions) 
were considered by Engelfriet and
Hoogeboom~\cite{eh01}. They showed that two-way
transducers and monadic-second order (MSO) definable transductions are
equivalent in the deterministic case (and even if the transduction is
single-valued, which is more general than determinism). This equivalence supports
thus the notion of ``regular''  functions, in the spirit of
classical results on regular word languages from automata theory 
and logics due to B\"uchi, Elgot,
Trakhtenbrot, Rabin, and others. 
A one-way transducer model that uses write-only registers as additional storage
was proposed a few years ago by Alur and Cern{\'y}~\cite{AlurCerny10}, 
and called streaming string transducer (SST).
SST were shown equivalent to two-way transducers and MSO definable 
transductions in the deterministic setting, and again, even in the 
single-valued case.

In the relational case the picture is less satisfactory, 
as expressive equivalence is only preserved for SST and non-deterministic
MSO transductions~\cite{ad11}, which extend the original MSO
transductions by existentially quantified monadic parameters. 
On the other hand, two-way transducers and 
SST are incomparable in the relational case. 
Between functions and relations there is however one
class of transductions that exhibits a better behavior, and this is
the class of finite-valued transductions. Being finite-valued means 
that there exists some constant $k$ such that every input belonging 
to the domain has at most $k$ outputs.

Finite-valued transductions were intensively studied in the setting of
one-way and two-way transducers. For one-way transducers,
$k$-valuedness can be checked in \PTIME~\cite{GI83}. In addition,
every $k$-valued one-way transducer can be effectively decomposed into
a union of $k$ unambiguous one-way transducers of exponential
size~\cite{web96,SS10}. For both two-way transducers and SST, checking
$k$-valuedness is in \PSPACE{}.

Besides expressiveness, another fundamental question concerning transducers 
is the equivalence problem, that is, the problem of deciding whether two transducers 
define the same relation (or the same partial function if we consider the single-valued case). 
The equivalence problem turns out to be \PSPACE-complete for 
deterministic two-way transducers~\cite{gur82}, single-valued two-way transducers,
as well as for single-valued SST~\cite{ad11}. 
For deterministic SST, equivalence is in \PSPACE~\cite{AC11popl}, but 
it is open whether this complexity upper bound is optimal. 
For arbitrary SST, and in fact even for non-deterministic 
one-way transducers over a unary output alphabet, 
equivalence is undecidable~\cite{FischerR68,iba78siam}. 
The equivalence problem for $k$-valued one-way transducers was shown to 
be decidable by Culik and Karhumäki using an elegant argument based on
Ehrenfeucht's conjecture~\cite{ck86}, and the authors noted that the
same proof goes through for two-way transducers as well. The
decidability status for the equivalence problem for $k$-valued SST was
first stated as an open problem in~\cite{ad11}. 
Another open problem is whether 
SST and two-way transducers are equivalent in the finite-valued case,
like in the single-valued case.
It is worth noting, however, that in the full
relational case SST and two-way transducers are incomparable.
Concerning this last open question, a partial positive answer was given
in \cite{gmps17}, by decomposing any finite-valued SST with only one 
register into a finite union of unambiguous SST. This decomposition 
result also entails the decidability of the equivalence problem
for the considered class.

The main result of this paper is a positive answer to the first 
question left open in~\cite{ad11}:

\begin{restatable}{theorem}{ThmMain}\label{thm:main}
The equivalence problem for finite-valued SST is decidable.
\end{restatable}

We show the above result with a proof idea due to Culik and
Karhumäki~\cite{ck86}, based on the Ehrenfeucht
conjecture. Our proof is much more involved, because SST
produce their outputs piece-wise, in contrast to one-way and
two-way transducers, that produce output linearly while reading the input. We
manage to overcome this obstacle using some (mild) word combinatorics
and word equations, by
introducing a suitable normalization procedure for SST. We believe that
our technique will also allow to solve the second problem left open in~\cite{ad11}, 
which is the expressive equivalence between finite-valued SST and two-way
transducers.

\smallskip
\paragraph*{Related work.} 
The equivalence problem for transducers has recently 
raised  interest
for more complex types of transducers in the single-valued case:
Filiot and Reynier showed that equivalence of copyful, deterministic
SST is decidable by showing them equivalent to HDT0L systems and
applying~\cite{ck86}, which contains the above-mentioned result as a
special case. Subsequently, Benedikt et
al.~showed that equivalence of copyful, deterministic
SST has Ackerman complexity, with a proof based on polynomial automata 
and ultimately on Hilbert's basis
theorem~\cite{BenediktDSW17}. Interestingly, the use of Hilbert's
basis theorem goes back to the proof of
Ehrenfeucht's conjecture~\cite{AL85,Guba86}.
A similar approach was used by Boiret et al.~in~\cite{BoiretPS18fsttcs}
to show that bottom-up register automata over unordered forests have a
decidable equivalence problem, see also the nice survey~\cite{boj19siglog}.

\smallskip
\paragraph*{Overview.} Section 2 introduces the transducer model, then Section 3 sets
up the technical machinery that allows to normalize finite-valued
SST. Section 4 shows the major normalization result, which holds for
left quotients of SST. Finally Section 5 recalls the Ehrenfeucht-based
proof for equivalence and the application to finite-valued SST.
A full version of the paper is available at \url{https://arxiv.org/abs/1902.06973}.

\section{Streaming string transducers}\label{sec:model}

A \emph{streaming string transducer} (\emph{SST}) is a tuple 
$\Tfull$, 
where
\begin{itemize}
\item $\Sigma$ and $\Gamma$ are finite input and output alphabets,
\item $X$ is a finite set of registers (usually denoted $x,x',x_1,x_2$, etc.),
\item $Q$ is a finite set of states,
\item $U$ is a finite set of register updates, 
      that is, functions from $X$ to $(X\uplus\Gamma)^*$,
\item $I,F\subseteq Q$ are subsets of states, defining the initial and final states,
\item $E\subseteq Q\times\Sigma\times U\times Q$ is a transition relation,
      describing, for each state and input symbol, the possible register 
      updates and target states,
\item $x_\out\in X$ is a register for the output.
\end{itemize}
Note that, compared to the original definition from \cite{AlurCerny10},
here we forbid for simplicity the use of final production rules,
that perform ad additional register update after the end of the input. 
This simplification is immaterial with respect to the decidability of the equivalence problem.
For example, it can be enforced, without loss of generality, by assuming that all 
well-formed inputs are terminated by a special marker, say $\dashv$, on which the 
transducer can apply a specific transition. We assume here that 
\emph{all inputs of a transducer are non-empty and of the form $u\dashv$, 
with $\dashv$ not occurring in $u$}.

Below, we recall briefly some key notions concerned with the computations of SST.

\paragraph*{Copyless restriction and capacity.}
An SST as above is \emph{copyless} if for all register updates $f\in U$, 
every register $x\in X$ appears at most once in the word $f(x_1)\dots f(x_m)$, 
where $X=\{x_1,\ldots,x_m\}$.
For a copyless SST, every output has length at most linear in 
the length of the input. More precisely, every output associated with
an input $u$ has length at most $c|u|$, where 
$c=\max_{f\in U}\sum_{x\in X}|f(x)|_\Gamma$ is
the maximum number of letters that the SST can add to its registers
along a single transition (this number $c$ is called \emph{capacity} of the SST).

\emph{Hereafter, we assume that all SST are copyless.}

\paragraph*{Register updates and flows.}
Every register update, and in general every function $f:X\rightarrow (X\uplus\Gamma)^*$ 
is naturally extended to a morphism on $(X\uplus\Gamma)^*$, by
defining it as identity over $\G$. 
When reasoning with register updates, it is sometimes possible to abstract away
the specific words over $\Gamma$, and only consider how the contents of the registers
flows into other registers. Formally, the \emph{flow} of an update 
$f:X\rightarrow(X\uplus\Gamma)^*$ is the bipartite graph that consists of 
two ordered sequences of nodes, one on the left and one on the right, 
with each node in a sequence corresponding to a specific register, 
and arrows that go from the node corresponding to register $x$ to
a right node corresponding to register $x$ whenever $x$ occurs in $f(x)$.
For example, the flow of the update $f$ defined by 
$f(x_1) = a \, x_1 \, aa \, x_3$, $f(x_2) = b\, a$, and $f(x_3) = x_2 \, b$
is the second bipartite graph in the figure on page~\pageref{example:flow}.

Note that there are finitely many flows on a fixed number of registers. 
Moreover, flows can be equipped with a natural composition operation:
given two flows $F_1$ and $F_2$, $F_1\cdot F_2$ is the bipartite graph obtained by glueing 
the right nodes of $F_1$ with the left nodes of $F_2$, and by shortcutting pairs of 
consecutive arrows.
We call \emph{flow monoid} of an SST $\T$ the monoid of flows generated 
by the updates of $\T$, with the composition operation as associative product.

\paragraph*{Transitions, runs, and loops.}
A transition $(q,a,f,q')$ of an SST $\T$ is conveniently denoted by the arrow
$q\trans{a / f}[\T] q'$, and the subscript $\T$ is often omitted when clear from 
the context.
A \emph{run} on $w=a_1\dots a_n$ is a sequence of transitions of the form 
\[
  q_0 \trans{a_1 / f_1}<\qquad> q_1 \trans{a_2 / f_2}<\qquad> \ldots \trans{a_n / f_n}<\qquad> q_n.
\]
Sometimes, a run as above is equally denoted by $q_0 \trans{w / f} q_n$,
so as to highlight the underlying input $w$ and the induced 
register update $f=f_1\circ \dots\circ f_n$.
A run is \emph{initial} (resp.~\emph{final}) if it begins
with an initial (resp.~final) state; it is \emph{successful}
if it is both initial and final. 

Given two registers $x,x'$ and a run $\r: \, q \trans{w/f} q'$,
we say that \emph{$x$ flows into $x'$ along $\r$} if $x$ occurs 
in $f(x')$. Note that this property depends only on the 
flow of the induced update $f$. 

An SST is said to be \emph{trimmed} is every state  occurs in at least
one successful run, so every state is reachable from the initial states and
co-reachable from the final states. This property can be easily 
enforced with a polynomial-time preprocessing.

When reasoning with automata, it is common practice
to use pumping arguments. 
Pumping will also be used here, but the notion
of loop needs to be refined as to take into account
the effect of register updates.
Formally, a \emph{loop} of a run $\rho$ of an SST is any 
non-empty factor of $\rho$ of the form $\gamma: \, q\trans{w/f}q$, 
that starts and ends in the same state $q$, and induces 
a \emph{flow-idempotent} update,
namely, an update $f$ such that $f$ and $f\circ f$ have the same flow.

\paragraph*{Outputs and finite-valuedness.}
The \emph{output} of a successful run $\r: \, q_0 \trans{w / f} q_n$ is defined 
as $\out(\r)= (f_0\circ f)(x_\out)$, where $f_0(x)=\emptystr$ for all $x\in X$. 
Sometimes, we write $\out(f)$ in place of $\out(\r)$.
The \emph{relation realized by an SST} is the set of pairs 
$(u,v)\in\Sigma^*\times\Gamma^*$, where $u$ is a well-formed
input (namely, terminating with $\dashv$) and $v$ is the output 
associated with some successful run on $u$.
An SST is \emph{$k$-valued} if 
for every input $u$, there 
are at most $k$ different outputs associated with $u$. 
It is \emph{single-valued} (resp.~\emph{finite-valued}) 
if it is $k$-valued for $k=1$ (resp.~for some $k\in\bbN$). The domain
an SST $\T$, denoted $\dom{\T}$, is the set of input words that have
some successful run in $\T$. Two SST $\T_1,\T_2$ are
\emph{equivalent}, denoted as $\T_1 \equiv \T
_2$, if they realize
the same relation over $\S^* \times \G^*$.

\paragraph*{Register valuations.}
A \emph{register valuation} is a function from $X$ to $\Gamma^*$. 
Given a successful run 
\[
  \rho: ~
  q_0 \trans{a_1 / f_1} q_1 \trans{a_2 / f_2}
  \ldots \trans{a_n / f_n} q_n
\]
and a position $i\in\{0,\dots,n\}$ in it, 
\emph{the register valuation at position $i$ in $\rho$} 
is the function $\val{\rho,i}$ that is defined inductively on $i$
as follows: $\val{\rho,0}(x)=\emptystr$, for all $x\in X$, 
and $\val{\rho,i+1} = \val{\rho,i}\circ f_i$.
Note that $\val{\rho,n}(x_\out)$ coincides with the final output 
$\out(f_1\circ\dots\circ f_n)$ produced by $\rho$.

Later we will generalize the notion of valuation to 
additional variables, called gaps.

\section{Normalizations}\label{sec:normalizations}

A major stumbling block in deciding equivalence of SST, 
as well as other crucial problems, lies in the fact that 
the same output can be produced by very different runs.
This phenomenon already appears with much simpler transducers,
e.g.~with one-way transducers, where runs may produce 
the same output, but at different speeds. 
However, the phenomenon is more subtle for SST, as the output is 
produced piece-wise, and not sequentially: runs with same output may 
appear to be different in many ways, e.g.~in terms of the flows 
of the register updates, or in terms of shifts of portions of the
output. 
The goal of this section is to provide suitable normalization 
steps that remove, one at a time, the above mentioned degrees 
of freedom in producing the same output. 

Another issue that we will be concerned with is the compatibility
of the normalization steps with constructions on transducers that
shortcut arbitrary long runs into a single transition. 
Essentially, we aim at having an effective notion of equivalence 
w.r.t.~final outputs that works not only for transitions but 
also for runs.

\paragraph*{Normalization of flows.}
In this section, 
\emph{$m$ will always denote the number of registers
of an SST and $X=\{x_1,\dots,x_m\}$ the set of registers.}
It is convenient to equip $X$ with a total order, say $x_1<\dots<x_m$.
Accordingly, we let $\chi=x_1\dots x_m$ be the juxtaposition of all
register names, and $f(\chi) = f(x_1)\dots f(x_m)$ for every register update $f$. 

We say that a register update $f$ is \emph{non-erasing} if for every register $x$,
$f(\chi)$ contains at least an occurrence of $x$ (in fact, exactly one, 
since $\T$ is copyless). This can be rephrased as a property of the flow of $f$,
where every node on the left must have an outgoing arrow.
In a similar way, we say that $f$ is \emph{non-permuting} if 
registers appear in $f(\chi)$ with their natural order 
and without jumps,
that is, $f(\chi) \in \Gamma^* x_1 \Gamma^* \, \dots \, \Gamma^* x_k \Gamma^*$, 
for some $k \le m$. As before, this can be rephrased by saying that the 
arrows in the flow of $f$ must not be crossing, 
and the target nodes to the right must form a prefix of $\chi$.
Below are some examples of updates with their flows: 
the first update $f$ is erasing, the second update $g$ is non-erasing 
but permuting, and the third update $h$ is non-erasing and non-permuting.
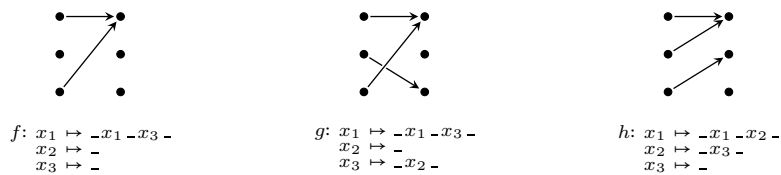
\label{example:flow}
\[
\begin{tikzpicture}[scale=1] 
\clip (-1,-2.1) |- (9.75,0.1) |- cycle;
\begin{scope}[xscale=0.8,yscale=0.5]
\draw (0,0) node [dot] (x) {};
\draw (0,-1) node [dot] (y) {};
\draw (0,-2) node [dot] (z) {};
\draw (1,0) node [dot] (x') {};
\draw (1,-1) node [dot] (y') {};
\draw (1,-2) node [dot] (z') {};
\draw [arrow] (x) to (x');
\draw [arrow] (z) to (x');
\draw (0.5,-3.5) node {$\begin{smallmatrix} 
                          f: ~   x_1 & \mapsto & \_\, x_1 \,\_\, x_3 \,\_ \\
                          \hfill x_2 & \mapsto & \_ \hfill \\
                          \hfill x_3 & \mapsto & \_ \hfill
                        \end{smallmatrix}$};
\end{scope}
\begin{scope}[xshift=4cm,xscale=0.8,yscale=0.5]
\draw (0,0) node [dot] (x) {};
\draw (0,-1) node [dot] (y) {};
\draw (0,-2) node [dot] (z) {};
\draw (1,0) node [dot] (x') {};
\draw (1,-1) node [dot] (y') {};
\draw (1,-2) node [dot] (z') {};
\draw [arrow] (x) to (x');
\draw [arrow] (y) to (z');
\draw [arrow] (z) to (x');
\draw (0.5,-3.5) node {$\begin{smallmatrix} 
                          g: ~   x_1 & \mapsto & \_\, x_1 \,\_\, x_3 \,\_ \\
                          \hfill x_2 & \mapsto & \_ \hfill \\
                          \hfill x_3 & \mapsto & \_\, x_2 \,\_ \hfill
                        \end{smallmatrix}$};
\end{scope}
\begin{scope}[xshift=8cm,xscale=0.8,yscale=0.5]
\draw (0,0) node [dot] (x) {};
\draw (0,-1) node [dot] (y) {};
\draw (0,-2) node [dot] (z) {};
\draw (1,0) node [dot] (x') {};
\draw (1,-1) node [dot] (y') {};
\draw (1,-2) node [dot] (z') {};
\draw [arrow] (x) to (x');
\draw [arrow] (y) to (x');
\draw [arrow] (z) to (y');
\draw (0.5,-3.5) node {$\begin{smallmatrix} 
                          h: ~   x_1 & \mapsto & \_\, x_1 \,\_\, x_2 \,\_ \\
                          \hfill x_2 & \mapsto & \_\, x_3 \,\_ \hfill \\
                          \hfill x_3 & \mapsto & \_ \hfill
                        \end{smallmatrix}$};
\end{scope}
\end{tikzpicture}
\]
We say that $\T$ is \emph{flow-normalized} if all its register 
updates are non-erasing and non-permuting.
Note that a flow-normalized SST with $m$ registers can have at 
most $\normflows$ different flows.

\begin{restatable}{proposition}{PropFlowNormalization}\label{prop:flow-normalization}
One can transform any SST into an equivalent flow-normalized one.
\end{restatable}

\begin{proof}
Let $\Tfull$ be an SST.
We need to construct an SST $\T'$ that simulates every run of $\T$ by guessing 
which registers in the current valuation contribute to form the final output,
and in which precise order, by appropriately modifying the register updates 
so as to enforce non-erasing and non-permuting behaviours. 

Formally, given any suffix $\rho: q\trans{u/f}q'$ of a successful run of $\T$,
we define a partial bijection $\pi_\rho:X\rightharpoonup X$ as follows:
for every register $x'$, if $x'$ is the $i$-th register occurring in 
$f(x_\out)$, then $\pi_\rho(x')=x_i$, otherwise,
if $x'$ does not occur in $f(x_\out)$, 
then $\pi_\rho$ is undefined on $x'$.
Any permutation of the form $\pi_\rho$ can be thought of as a renaming 
of registers that contribute to the final output. By construction, the 
range of such a renaming is always an initial interval of the registers, 
i.e.~$\rng{\pi_\rho}=\{x_1,\ldots,x_k\}$ for some $k\le m$.
For the sake of brevity, hereafter we call \emph{renaming} any function
of the above form, that is, any bijection from a subset $\{x_{i_1},\dots,x_{i_k}\}$ 
of $X$ to $\{x_1,\dots,x_k\}$, for $0\le k\le m$. We also let 
$\dlen{\pi} = |\dom{\pi}|$ ($= |\rng{\pi}|$)
for any renaming $\pi$.

The normalized SST is defined as $\Tfulln[T']<Q'><U'><I'><E'><F'>$, where:
\begin{itemize}
\item $Q'=Q\times R$, where $R$ is the set of all renamings,
\item $U'$ contains all {\sl non-erasing and non-permuting} updates of the form 
      $f[\pi\rightarrow \pi']$, for $f\in U$, $\pi,\pi'\in R$, where 
      $\pi$ is defined precisely on those registers that occur in
      $f\circ (\pi')^{-1}(\{1,\ldots,\dlen{\pi'}\})$,
      and $f[\pi\rightarrow \pi']$ is defined by 
      \[
        f[\pi\rightarrow \pi'] (x_i) = 
        \begin{cases} 
          \pi\circ f\circ (\pi')^{-1} (x_i) & \text{if $i \le \dlen{\pi'}$} \\
          x_{i+\dlen{\pi}-\dlen{\pi'}}      & \text{if $i > \dlen{\pi'}$ 
                                                    and $i+\dlen{\pi}-\dlen{\pi'} \le m$}, \\
          \emptystr                         & \text{if $i>\dlen{\pi'}$ 
                                                    and $i+\dlen{\pi}-\dlen{\pi'} > m$},
        \end{cases}
      \]
\item $I'=I\times R$, 
\item $E'$ contains all transition rules of the form $(q,\pi)\trans{a/ f'} (q',\pi')$,
      with $q\trans{a/ f}q'$ transition rule in $E$ and $f'=f[\pi\rightarrow \pi']$,
\item $F'=F\times\{\pi_\out\}$, where $\pi_\out$ is the renaming defined only
      on $x_\out$ and mapping it to $x_1$.
\end{itemize}
It is routine to show that $\T'$ is flow-normalized and equivalent to $\T$.
\end{proof}

\medskip
Recall that a register valuation is a function from $X$ to $\Gamma^*$. 
With a flow-normalized SST, one can also define a dual notion of valuation, 
representing `gaps' between registers that shrink along the run. 
For this we introduce $m+1$ fresh variables $y_0,y_1,\ldots,y_m$, 
called \emph{gaps}. 
Hereafter, \emph{$Y=\{y_0,y_1,\dots,y_m\}$ will always denote the set of gaps}.
We use the term \emph{valuation} to generically denote a
register/gap valuation, that is, a function from $X\uplus Y$ to $\Gamma^*$.

The idea is that a gap $y_j$ represents a word that is inserted 
between register $x_j$ (if $j>0$) and register $x_{j+1}$ (if $j<n$)
so as to form the final output.
Formally, given a word $w\in \Gamma^* x_1 \Gamma^*  \dots \Gamma^* x_k \Gamma^*$, 
with $k\le m$, and given two registers $x_i,x_j$, with $i<j$,
we denote by $w\gap{x_i,x_j}$ the maximal factor of $w$ strictly between 
the unique occurrence of $x_i$ and the unique occurrence of $x_j$, 
using the following conventions for the degenerate cases:
if $i=0$, then $w\gap{x_i,x_j}$ is a maximal prefix of $w$; 
if $i>0$ but there is no occurrence of $x_i$, then $w\gap{x_i,x_j}=\emptystr$;
finally, if there is an occurrence of $x_i$ but no occurrence of $x_j$ in $w$, 
then $w\gap{x_i,x_j}$ is a maximal suffix.
Given a run 
\[
  \rho: ~
  q_0 \trans{a_1 / f_1} q_1 \trans{a_2 / f_2} \ldots \trans{a_n / f_n} q_n
\]
and a position $i$ in it, 
\emph{the valuation at position $i$ of $\rho$} is the function 
$\val{\rho,i}:X\uplus Y\rightarrow\Gamma^*$ such that 
\begin{itemize}
  \item $\val{\rho,i}$ restricted to $X$
        is the register valuation at position $i$ of $\rho$,
  \item $\val{\rho,i}$ maps every gap $y_j$ to the word 
        $(f_{i+1}\circ\dots\circ f_n)(\chi)\gap{x_j,x_{j+1}}$.
\end{itemize}

By definition, the image of the word
$\zeta = y_0 \, x_1 \, y_1 \, \dots \, x_m \, y_m$ via the valuation $\val{\rho,i}$
is always equal to the final output $\out(\rho)$, for all positions $i$.
In this sense, the sequence of valuations $\val{\rho,0},\val{\rho,1},\ldots,\val{\rho,n}$ 
can be identified with a sequence of factorizations of $\out(\rho)$.
For example, below are the factorizations of the output
before and after a transition with register update $f$ 
such that $f(x_1)=s\,x_1\,u\,x_2\,t$ and $f(x_2)=v$,
for $s,u,t,v\in\Gamma^*$:
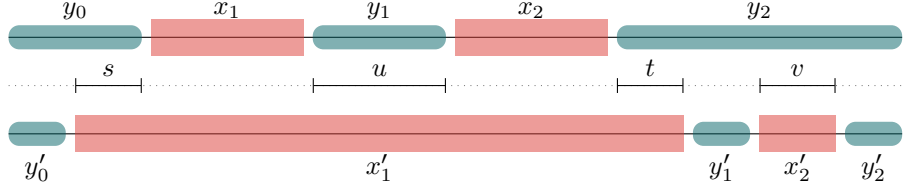
\label{example:register-gap-update}
\[
\begin{tikzpicture}[yscale=.8]
\clip (0,-2.4) |- (12,0.6) |- cycle;
\begin{scope}[yscale=0.2]
\draw (0,0) node (A) {};
\draw (2,0) node (B) {};
\draw (4,0) node (C) {};
\draw (6,0) node (D) {};
\draw (8,0) node (E) {};
\draw (12,0) node (F) {};
\draw (A) -- (F);
\draw [draw=none,fill=nicered,opacity=0.5] ([yshift=-1.5cm]B) rectangle ([yshift=1.5cm]C);
\draw [draw=none,fill=nicered,opacity=0.5] ([yshift=-1.5cm]D) rectangle ([yshift=1.5cm]E);
\draw [draw=none,fill=nicecyan,opacity=0.5,rounded corners] 
      ([yshift=-1cm]A.east) rectangle ([yshift=1cm]B.west);
\draw [draw=none,fill=nicecyan,opacity=0.5,rounded corners] 
      ([yshift=-1cm]C.east) rectangle ([yshift=1cm]D.west);
\draw [draw=none,fill=nicecyan,opacity=0.5,rounded corners] 
      ([yshift=-1cm]E.east) rectangle ([yshift=1cm]F.west);
\draw ([yshift=2.2cm]$(A)!0.5!(B)$) node {$y_0$};
\draw ([yshift=2.2cm]$(B)!0.5!(C)$) node {$x_1$};
\draw ([yshift=2.2cm]$(C)!0.5!(D)$) node {$y_1$};
\draw ([yshift=2.2cm]$(D)!0.5!(E)$) node {$x_2$};
\draw ([yshift=2.2cm]$(E)!0.5!(F)$) node {$y_2$};
\end{scope}
\begin{scope}[yshift=-0.8cm,yscale=0.2]
\draw (0,0) node (A) {};
\draw (1,0) node (B) {};
\draw (2,0) node (C) {};
\draw (4,0) node (D) {};
\draw (6,0) node (E) {};
\draw (8,0) node (F) {};
\draw (9,0) node (G) {};
\draw (10,0) node (H) {};
\draw (11,0) node (I) {};
\draw (12,0) node (J) {};
\draw [dotted,gray] (A) -- (J);
\draw (B.center) edge [|-|] node [above] {$s$} (C.west);
\draw (D.east) edge [|-|] node [above] {$u$} (E.west);
\draw (F.east) edge [|-|] node [above] {$t$} (G.center);
\draw (H.center) edge [|-|] node [above] {$v$} (I.center);
\end{scope}
\begin{scope}[yshift=-1.6cm,yscale=0.2]
\draw (0,0) node (A) {};
\draw (1,0) node (B) {};
\draw (9,0) node (C) {};
\draw (10,0) node (D) {};
\draw (11,0) node (E) {};
\draw (12,0) node (F) {};
\draw (A) -- (F);
\draw [draw=none,fill=nicered,opacity=0.5] ([yshift=-1.5cm]B) rectangle ([yshift=1.5cm]C);
\draw [draw=none,fill=nicered,opacity=0.5] ([yshift=-1.5cm]D) rectangle ([yshift=1.5cm]E);
\draw [draw=none,fill=nicecyan,opacity=0.5,rounded corners] 
      ([yshift=-1cm]A.east) rectangle ([yshift=1cm]B.west);
\draw [draw=none,fill=nicecyan,opacity=0.5,rounded corners] 
      ([yshift=-1cm]C.east) rectangle ([yshift=1cm]D.west);
\draw [draw=none,fill=nicecyan,opacity=0.5,rounded corners] 
      ([yshift=-1cm]E.east) rectangle ([yshift=1cm]F.west);
\draw ([yshift=-2.8cm]$(A)!0.5!(B)$) node {$y'_0$};
\draw ([yshift=-2.8cm]$(B)!0.5!(C)$) node {$x'_1$};
\draw ([yshift=-2.8cm]$(C)!0.5!(D)$) node {$y'_1$};
\draw ([yshift=-2.8cm]$(D)!0.5!(E)$) node {$x'_2$};
\draw ([yshift=-2.8cm]$(E)!0.5!(F)$) node {$y'_2$};
\end{scope}
\end{tikzpicture}
\]
This also suggests the principle that gaps, like registers,
are updated along transitions via suitable morphisms, but in a symmetric 
way, that is, from right to left.
For instance, in the above picture, the gaps $y_0,y_1,y_2$ are updated 
by the function $\dual{f}$ such that $\dual{f}(y_0) = y_0\,s$, $\dual{f}(y_1)=u$, and
$\dual{f}(y_2)=t\,y_1\,v\,y_2$. In general, the function $\dual{f}$, called
\emph{gap update}, is uniquely determined by the register update $f$, and
vice versa, $f$ is uniquely determined by the gap update $\dual{f}$. 
Another perhaps interesting phenomenon is that the gap update $\dual{f}$ is also 
non-erasing and non-permuting (the notion of non-permuting gap assignment
is defined w.r.t.~the reverse order $y_m<\dots<y_0$).

\paragraph*{Normalization of states.}
The next normalization step splits the states of an
SST in such a way that it becomes possible to associate 
with each state an over-approximation of the possible 
register/gap valuations witnessed when the state is 
visited along a successful run. 
These over-approximations are very simple languages over 
the output alphabet $\Gamma$, e.g.~singleton languages 
like $\{aba\}$ and periodic languages like $\{ab\}^*\{a\}$
(often denoted $(ab)^*a$ to improve readability).
Basically our over-approximations refer to length and period
constraints.
The \emph{period} of a word $w$ is the least number $0<p\le |w|$ 
such that $w$ is a prefix of $(w[1,p])^\omega$. For example, the
period of $w=abcab$ is $3$. 

For a given parameter $\alpha\in\bbN$ we define the family $\cL_\alpha$ that contains: 
\begin{itemize}
\item the empty language $\emptyset$,
\item the singleton languages $\{u\}$, 
      with $u\in\Gamma^*$ and $|u|\le \alpha$,
\item the periodic languages $u^* v$, 
      with $u\in\Gamma^+$ \emph{primitive}
      (i.e.~$u=w^k$ only if $k=1$), 
      \par\noindent
      $|u|\le\alpha$, and $v\in\Gamma^*$
      strict prefix of $u$,
\item the universal language $\Gamma^*$.
\end{itemize}

\noindent
The languages in $\cL_\alpha$, partially ordered by containment, form a finite 
meet semi-lattice, where the meet is the intersection $\cap$.
We depict here part of the lattice $\cL_\alpha$ for a parameter $\alpha\ge 3$:
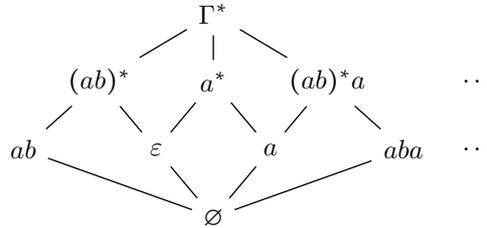
\[
\begin{tikzpicture}[yscale=0.9]
\draw (0,0) node (bot) {$\emptyset$};
\draw (-2.5,1) node (ab) {$ab$};
\draw (-0.75,1) node (empty) {$\emptystr$};
\draw (0.75,1) node (a) {$a$};
\draw (2.5,1) node (aba) {$aba$};
\draw (3.5,1) node {$\dots$};
\draw (-1.5,2) node (abab) {$(ab)^*$};
\draw (0,2) node (aa) {$a^*$};
\draw (1.5,2) node (ababa) {$(ab)^*a$};
\draw (3.5,2) node {$\dots$};
\draw (0,3) node (top) {$\Gamma^*$};
\draw [link] (bot) to (empty);
\draw [link] (bot) to (a);
\draw [link] (bot) to (ab);
\draw [link] (bot) to (aba);
\draw [link] (empty) to (abab);
\draw [link] (ab) to (abab);
\draw [link] (empty) to (aa);
\draw [link] (a) to (ababa);
\draw [link] (a) to (aa);
\draw [link] (aba) to (ababa);
\draw [link] (abab) to (top);
\draw [link] (aa) to (top);
\draw [link] (ababa) to (top);
\end{tikzpicture}
\]

The semi-lattice structure allows to derive a best over-approximation in $\cL_\alpha$ 
of any language $L\subseteq\Gamma^*$, that is:
$\closure L = \bigcap \{ L'\in\cL_\alpha \mid L'\supseteq L\}$.
We will mostly use the approximation operator $\closure{}$ on singleton languages.
For example, for $\alpha=3$, we have $\closure{\{aba\}}=\{aba\}$, 
$\closure{\{ababa\}}=(ab)^*a$, and
$\closure{\{abbb\}}=\Gamma^*$. Note also that if $|w| \le \a$ then $\closure{w}=\set{w}$.
A useful property is the compatibility of $\closure{}$ with concatenation,
which immediately extends to compatibility with word morphisms:

\begin{restatable}{lemma}{LemApproximatedConcat}\label{lem:approximated-concat}
$\closure{(L_1\cdot L_2)} = \closure{(\closure{L_1}\cdot\closure{L_2})}$
for every $\alpha\in\bbN$ and $L_1,L_2\subseteq\Gamma^*$.
\end{restatable}

\begin{proof}
The left-to-right containment follows easily by monotonicity of $\closure{}$.
The converse containment boils down to proving that for every
$L\in\cL_\alpha$ and $w\in\Gamma^*$, 
$L\supseteq \{w\} \cdot L_2$ implies $L\supseteq \{w\}\cdot\closure{L_2}$
(one can then take the conjunction of the latter implication over all $w\in L_1$, 
and prove in this way that $L\supseteq L_1\cdot L_2$ implies
$L\supseteq L_1\cdot\closure{L_2}$, finally, using symmetric arguments,
one derives that $L\supseteq L_1\cdot L_2$ implies
$L\supseteq \closure{L_1}\cdot\closure{L_2}$).

If $L$ is empty, a singleton, or the universal language $\Gamma^*$, or
if $L_2$ is empty, then the considered implication holds trivially. 
So, we consider the case where $L$ is a periodic language 
of the form $u^* v$, with $u$ primitive and $v$ prefix of $u$, 
and $L_2$ is non-empty.
Since $L$ contains at least one word with $w$ as prefix, we know that
$w\in u^* v'$, for some $v'$ prefix of $u$.
Similarly, for every word $w'\in L_2$, $L$ must contain at least one word
with $w'$ as suffix, and hence $L_2\subseteq v'' \,u^* v$ for some $v''$ 
suffix of $u$ such that $v' \, v''=u$. 
This implies $\closure{L_2}\subseteq v''\, u^* v$,
and hence 
\[
  L ~=~ u^* v 
    ~\supseteq~ u^* v' \, v''\, u^*v 
    ~\supseteq~ \{w\}\cdot\closure{L_2}.
\qedhere
\]
\end{proof}

\medskip
Recall that $X,Y$ denote, respectively, the sets of registers and gaps
of a flow-normalized SST.
Given a valuation $\nu: X\uplus Y \rightarrow \Gamma^*$, 
its \emph{$\alpha$-approximant} is the function 
$\closure{\nu}: X\uplus Y\rightarrow\cL_\alpha$ 
that maps any $z\in X\uplus Y$ to the language
$\closure{\{\nu(z)\}}$. The set of $\alpha$-approximants is 
denoted $\cL_\alpha^{ X\uplus Y}$, and consists of all maps from 
$X\uplus Y$ to $\cL_\alpha$.
Further let
\[
  \Val{q} = \{ \val{\rho,i} \mid \text{$\rho$ successful run visiting $q$ at any position $i$} \}
\]
be the set of possible valuations induced by an arbitrary 
successful run when visiting state $q$.

A first desirable property is that all valuations
in $\Val{q}$ have the {\sl same} $\alpha$-approximant, 
which is thus determined by the state $q$. 
Formally, given a flow-normalized SST $\T$ with trimmed state space $Q$,
we say that $\T$ \emph{admits $\alpha$-approximants} 
if every state $q\in Q$ can be effectively annotated with 
an $\alpha$-approximant $\App_q\in\cL_{\alpha}^{X\uplus Y}$ 
in such a way that
\begin{equation}\label{eq:approximants}
  \forall \nu\in\Val{q}\, : \qquad \closure{\nu} = \App_q.
\end{equation}
This condition is best understood as an invariant on lengths and 
periods that can be enforced on valuations of registers and gaps when 
visiting a particular state.
For instance, when the approximant $\App_q$ guarantees a certain period, 
then this period will be the same for all valuations occurring
at $q$, independently of the specific initial run that may lead 
to $q$ (for registers), and of the run that may lead from $q$ 
to an accepting state (for gaps). 

The proposition below shows that it is always possible to 
refine any SST $\T$ so as to admit $\a$-approximants, 
for any parameter $\alpha$. 
The proof for $\a$-approximants that concern only registers 
could be understood as unfolding the SST $\T$, and merging nodes
corresponding to any two inputs $u$ and $v$, with $u$ prefix 
of $v$, whenever the induced $\a$-approximants at $u$ and $v$
are the same for every register $x$. 
In general, the resulting SST can be seen a covering 
of the original SST $\T$,
in the sense formalized by Sakarovitch and de Souza in \cite{SS08}: 
$\T'$ is a \emph{covering} of $\T$ if the
states of $\T'$ can be mapped homomorphically to states 
of $\T$, while preserving transitions and the distinction into
initial and final states, and, moreover, the outgoing transitions of
every state of $\T'$ map one-to-one to outgoing transitions
of a corresponding state of $\T$. This implies that the
successful runs of $\T$ and those of $\T'$ are in one-to-one
correspondence. 

\begin{restatable}{proposition}{PropApproximants}\label{prop:approximants}
Let $\T$ be a flow-normalized SST, and let $\alpha\in\bbN$.
One can construct an equivalent flow-normalized SST $\T'$ 
that admits $\alpha$-approximants
and that is a covering of $\T$.
\end{restatable}

\begin{proof}
Let $\Tfulln$ be a flow-normalized SST with a trimmed state space $Q$,
and let $\alpha\in\bbN$.
The desired SST that admits $\alpha$-approximants is defined as 
$\Tfulln[T']<Q'><U><I'><E'><F'>$, where
\begin{itemize}
\item $Q'=Q\times \cL_\alpha^X \times \cL_\alpha^Y$ 
      --- namely, the states of $\T'$ are obtained by 
      annotating the states of $\T$ with $\alpha$-approximants
      of register valuations and gap valuations,
\item $I'=I\times \{\App_\emptystr\} \times \cL_\alpha$,
      with 
      $\App_\emptystr(x) = \emptystr$ for all $x \in X$
      --- namely, the initial states of $\T'$ have 
      $\alpha$-approximants for register valuations 
      initialized with the empty word $\emptystr$, 
\item $F'=F\times \cL_\alpha^X \times \{\App_\emptystr\}$,
      with $\App_\emptystr(y)=\emptystr$ for all $y\in Y$
      --- namely, the final states of $\T'$ have
      $\alpha$-approximants for gap valuations 
      initialized with $\emptystr$,
\item $E'$ consists of transitions of the form 
      $(q,\App_X,\App_Y) \trans{a/ f} (q',\App'_X,\App'_Y)$,
      where $q\trans{a/ f} q'$ is a transition in $E$,
      $\App'_X(x)=\closure{\big(\App_X(f(x))\big)}$ for all registers $x\in X$,
      and $\App_Y(y)=\closure{\big(\App'_Y(\dual{f}(y))\big)}$ for all gaps $y\in Y$,
      with $\dual{f}$ gap update determined by $f$.
      Here, $\App_X(f(x))$ means
      substituting the languages from $\cL_\a$ associated with
      registers $x'\ \in X$ into $f(x)$, and similarly for
      $\dual{f}(y)$.
      Intuitively, the approximation $\App'_X(x)$ of a target register 
      valuation is obtained by considering the effect of the update $f$ 
      on the register $x$ when the source valuation ranges over $\App_X(x)$, 
      and symmetrically for a gap $y$. 
\end{itemize}
It is clear from the above definitions that $\T'$ is a covering
of $\T$.
This implies that the successful runs of $\T$ are precisely
the successful runs of $\T'$ devoid of the $\alpha$-approximants,
and hence $\T'$ is flow-normalized and equivalent to $\T$.

It remains to prove that $\T'$ admits $\alpha$-approximants. 
This boils down to proving that for every successful run $\rho$ of $\T'$ 
that visits a state $(q,\App_X,\App_Y)$ at position $i$, we have 
$\App_X=(\closure{\val{\rho,i}})|_X$ and $\App_Y=(\closure{\val{\rho,i}})|_Y$,
where $|_X$ (resp.~$|_Y$) denotes the restriction of a function 
to the set $X$ (resp.~$Y$).
For this, we recall that if $\val{\rho,i}$ and $\val{\rho,i+1}$ 
are two consecutive valuations w.r.t.~a register update $f$, then 
\[
  \val{\rho,i+1}|_X ~=~ \val{\rho,i}|_X\circ f
  \qquad\qquad\text{and}\qquad\qquad
  \val{\rho,i}|_Y ~=~ \val{\rho,i+1}|_Y \circ \dual{f}.
\]
By Lemma \ref{lem:approximated-concat}, we derive 
\[
  (\closure{\val{\rho,i+1}})|_X 
  ~=~ \closure{\big((\closure{\val{\rho,i}})|_X\circ f\big)}
  \qquad\text{and}\qquad 
  (\closure{\val{\rho,i}})|_Y 
  ~=~ \closure{\big((\closure{\val{\rho,i+1}})|_Y\circ \dual{f}\big)}.
\]
Thanks to this, using simple inductions on $i$, one can verify that 
$\App|_X = \closure{(\val{\rho,i}})|_X$, and symmetrically that
$\App|_Y = \closure{(\val{\rho,i}})|_Y$. 
\end{proof}

\noindent
\emph{Notation.} Whenever an SST admits $\alpha$-approximants $\App_q$ as above, 
it is convenient to denote its states by triples of the form $(q,\App_X,\App_Y)$,
where $\App_X$ (resp.~$\App_Y$) is the restriction of the $\a$-approximant $\App_q$ of state $q$ to registers (resp.~gaps).

\medskip
Note that the smaller the parameter $\alpha$, the weaker is the 
property required for $\a$-approximants (in particular, for $\alpha=0$ the lattice
$\cL_\a$ collapses to $\emptyset$ and $\G^*$). 
Choosing $\alpha$ to be 
at least the capacity of the SST is already a reasonable choice, 
as it gives a nice characterization of equivalence of transitions 
w.r.t.~the produced outputs (cf.~Lemma~\ref{lem:equivalence-by-approximants} 
below).
However, we will see that it is desirable to have  
even finer approximants, in such a way that our results
will be compatible with left quotients of SST, that shortcut 
arbitrary long runs into single transitions. 
We postpone the technical details to Section \ref{sec:shortcut}, 
and only provide a rough intuition underlying the choice of the 
appropriate parameter $\alpha$. 
We will choose $\alpha$ much larger than the capacity of the SST, 
so that, by pumping arguments, one can show that,
for every state $q$ and every parameter $\beta\ge\alpha$, 
the $\beta$-approximant cannot be strictly smaller 
than the $\alpha$-approximant on {\sl all} valuations 
from $\Val{q}$.

\paragraph*{Normalization of transitions.}
We finally turn to studying a notion of equivalence on transitions 
that is similar to the two-sided Myhill-Nerode equivalence on words. 
We will only compare transitions that consume the same input letter
and link the same pair of states. Instead of using words as 
two-sided contexts, we will use initial and final runs that 
can be attached to the considered transitions in order to form 
successful runs, and instead of comparing membership in a language, 
we will compare the effect on the produced outputs. 

Consider two transitions $\t_1:\, q\trans{a/f_1}q'$ and $\t_2:\,
q\trans{a/f_2}q'$. 
We say that $\t_1$ and $\t_2$ are \emph{equivalent} 
if for every initial run $\r$ leading to $q$ and every final run $\s$ starting in $q'$,
the outputs $\out(\r \, \t_1 \, \s)$ and $\out(\r \, \t_2 \, \s)$ are
equal. We often refer to $(\r,\s)$ as a \emph{context} for $\t_1,\t_2$.

\medskip
In general, two transitions of an SST having the same source, target
and $\S$-label might turn out to be non-equivalent, 
and still produce the same output within specific contexts. 
However, Lemma~\ref{lem:equivalence-by-approximants} below shows that this is not the case with 
$\alpha$-approximants at hand, provided that $\alpha$ 
is at least the capacity of the SST. 
More precisely, we will show that the equivalence of two transitions 
$\t_1:\, (q,\App_X,\App_Y)\trans{a/f_1}(q',\App'_X,\App'_Y)$ and 
$\t_2:\, (q,\App_X,\App_Y)\trans{a/f_2}(q',\App'_X,\App'_Y)$,
where $f_1,f_2$ have the same flow, only depends on the 
$\alpha$-approximants $\App_X,\App'_Y$ that annotate 
the source and target states. This will imply that $\t_1,\t_2$
either always produce the same output or always produce different outputs,
independently of the surrounding contexts.
To prove the statement, we have to consider register valuations 
induced by initial runs, and symmetrically gap valuations 
induced by final runs. It helps to introduce the following:

\smallskip

\noindent
\emph{Notation.} 
Given an initial run $\r$, $\val{\r\pos}$ is the register valuation 
induced at the end of $\r$; symmetrically, $\val{\pos\s}$ is the gap 
valuation induced at the beginning of a final run $\s$.

\medskip

We also recall a consequence of the flow normalization: 
the effect on the final output of an update $f$ 
that occurs in a successful run can be described 
by a word $\eff{f}$ over the alphabet $X \uplus Y \uplus \G$, defined
as $\eff{f}  = y_0\, f(x_1)\, y_1 \dots f(x_m)\, y_{m+1}$.
Note that in $\eff{f}$ each register (resp.~gap) occurs exactly once, according to
the order $x_1 < \dots < x_m$ (resp.~$y_0 < \dots < y_m$)
--- the occurrences of registers and gaps, however, may not be strictly interleaved. 
In $\eff{f}$, an occurrence of $x_i\in X$ represents an abstract 
valuation for register $x_i$ {\sl before} applying the update $f$, 
while an occurrence of $y_j\in Y$ represents an abstract valuation 
for the gap $y_j$ {\sl after} applying $f$.
In particular, note that the $x$'s and the $y$'s refer 
to valuations induced at different positions of a run.
The maximal factors of $\eff{f}$ that are entirely over $\G$ 
represent the words that need to be added in order to get 
the register valuation {\sl after} $f$, or equally the gap 
valuation {\sl before} $f$. 
For instance, by reusing the example update $f$ 
from page \pageref{example:register-gap-update},
where $f(x_1)=s\,x_1\,u\,x_2\,t$ and $f(x_2)=v$,
the effect of $f$ is described by the word
$\eff{f} = y_0 \, s \, x_1 \, u \, x_2 \, t \, y_1 \, v \, y_2$,
suggestively depicted as
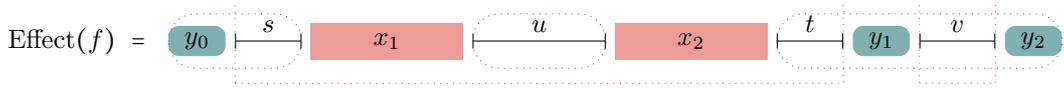
\[
\begin{tikzpicture}[yscale=.8] 
\begin{scope}[yscale=0.2]
\clip (-2,-3.6) |- (11.9,3.5) |- cycle;
\draw (0,0) node [left] {$\eff{f} ~=\,$};
\draw (0,0) node (A) {};
\draw (1,0) node (B) {};
\draw (2,0) node (C) {};
\draw (4,0) node (D) {};
\draw (6,0) node (E) {};
\draw (8,0) node (F) {};
\draw (9,0) node (G) {};
\draw (10,0) node (H) {};
\draw (11,0) node (I) {};
\draw (12,0) node (J) {};
\draw [draw=none,fill=nicered,opacity=0.5] ([yshift=-1.5cm]C) rectangle ([yshift=1.5cm]D);
\draw [draw=none,fill=nicered,opacity=0.5] ([yshift=-1.5cm]E) rectangle ([yshift=1.5cm]F);
\draw [draw=none,fill=nicecyan,opacity=0.5,rounded corners] 
      ([yshift=-1.25cm]A.east) rectangle ([yshift=1.25cm]B.west);
\draw [draw=none,fill=nicecyan,opacity=0.5,rounded corners] 
      ([yshift=-1.25cm]G.east) rectangle ([yshift=1.25cm]H.west);
\draw [draw=none,fill=nicecyan,opacity=0.5,rounded corners] 
      ([yshift=-1.25cm]I.east) rectangle ([yshift=1.25cm]J.west);
\draw ($(A)!0.5!(B)$) node {$y_0$};
\draw (B.center) edge [|-|] node [above] {$s$} (C.west);
\draw ($(C)!0.5!(D)$) node {$x_1$};
\draw (D.east) edge [|-|] node [above] {$u$} (E.west);
\draw ($(G)!0.5!(H)$) node {$y_1$};
\draw (F.east) edge [|-|] node [above] {$t$} (G.center);
\draw ($(E)!0.5!(F)$) node {$x_2$};
\draw (H.center) edge [|-|] node [above] {$v$} (I.center);
\draw ($(I)!0.5!(J)$) node {$y_2$};
\draw [dotted,nicered,opacity=0.75] ([yshift=-3.5cm]B) rectangle ([yshift=3.5cm]G);
\draw [dotted,nicered,opacity=0.75] ([yshift=-3.5cm]H) rectangle ([yshift=3.5cm]I);
\draw [dotted,nicecyan,opacity=0.75,rounded corners=10] ([yshift=-2.25cm]A.east) rectangle ([yshift=2.25cm]C.west);
\draw [dotted,nicecyan,opacity=0.75,rounded corners=10] ([yshift=-2.25cm]D.east) rectangle ([yshift=2.25cm]E.west);
\draw [dotted,nicecyan,opacity=0.75,rounded corners=10] ([yshift=-2.25cm]F.east) rectangle ([yshift=2.25cm]J.west);
\end{scope}
\end{tikzpicture}
\]
(here the lengths of the blocks labeled with variables are immaterial).
Note that the $y_j$ above are in fact the $y'_j$ from the picture
at page \pageref{example:register-gap-update}.  
In the above figure we have also highlighted with dotted rectangles the 
factors that represent gap valuations before the update (e.g.~$y_0\,s$), 
and register valuations after the update (e.g.~$s\,x_1\,u\,x_2\,t$).

Given a valuation $\nu$ and two approximants $\App\in\cL_\alpha^X$ 
and $\App'\in\cL_\alpha^Y$, one for register valuations and the other
for gap valuations, we write $\nu\in\App\uplus\App'$ to mean that
$\nu(x)\in\App(x)$ and $\nu(y)\in\App'(y)$ for all $x\in X$ and $y\in Y$.

\begin{restatable}{lemma}{LemEquivalenceByApproximants}\label{lem:equivalence-by-approximants}
Let $\T$ be a trimmed flow-normalized SST.
Given two transitions 
$\t_i : \, q \trans{a/f_i} q'$, with $i \in \set{1,2}$,
a context $(\r,\s)$ for them, and the 
$\alpha$-approximants $\App = (\closure{(\val{\r\pos})})|_X$ and 
$\App' = (\closure{(\val{\pos\s})})|_Y$, with $\alpha\in\bbN$,
the following holds:
\begin{enumerate}
\item If $\eff{f_1} = \eff{f_2}$ holds
      on all valuations $\nu\in\App\uplus\App'$, 
      then $\out(\r\,\t_1\,\s) = \out(\r\,\t_2\,\s)$.
\item If $\t_1,\t_2$ have the same flow,
      $\out(\r\,\t_1\,\s) = \out(\r\,\t_2\,\s)$,
      and $\alpha\ge c$, where $c$ is the capacity of $\T$,
      then $\eff{f_1} = \eff{f_2}$ holds on all valuations 
      $\nu\in\App\uplus\App'$.
\end{enumerate}
\end{restatable}

\begin{proof}
The proof exploits the fact that, since $\T$ is flow-normalized,
for every successful run $\r\,\t_i\,\s$, for both $i=1$ and $i=2$, the 
substitution in $\eff{f_i}$ of every variable $x\in X$ (resp.~$y\in Y$) 
with the word $\val{\r\pos}(x)$ (resp.~$\val{\pos\s}(y)$) gives 
precisely the output $\out(\r\,\t_i\,\s)$.
This implies that
\begin{align*}
  \out(\r\,\t_1\,\s) = \out(\r\,\t_2\,\s)
  \quad\text{if and only if}\quad
  (\val{\r\pos}) \uplus (\val{\pos\s})
  \,\sat\, \eff{f_1} = \eff{f_2}
\tag{$\star$}
\end{align*}

We prove the first claim, which holds 
for any arbitrary parameter $\alpha\in\bbN$.
By construction, we have $\val{\r\pos}(x) \in \App(x)$, 
for all $x\in X$, and $\val{\pos\s}(y) \in \App'(y)$, for all $y\in Y$.
The previous property ($\star$) immediately implies that 
$\tau_1$ and $\tau_2$ produce the same output within the context 
$(\r,\s)$ {\sl if} the equation $\eff{f_1} = \eff{f_2}$ holds for 
all valuations $\nu\in\App \uplus \App'$.

To prove the second claim, we assume that $\t_1$ and $\t_2$
have the same flow, we fix a context $(\r,\s)$ for them,
and we let $\App = \closure{(\val{\r\pos})}$ 
and $\App' = \closure{(\val{\pos\s})}$ for some $\alpha\ge c$,
where $c$ is the capacity of $\T$.
We need to prove that the equation $\eff{f_1} = \eff{f_2}$ 
holds for all valuations $\nu\in\App\uplus\App'$.
We proceed by equating the two words $\eff{f_1}$ and $\eff{f_2}$, 
and we study how the various blocks inside $\eff{f_1}$ and $\eff{f_2}$
(i.e.~variables and maximal factors over $\G$) are aligned.
The reader may refer to the figure below, which gives an example
of possible alignments:
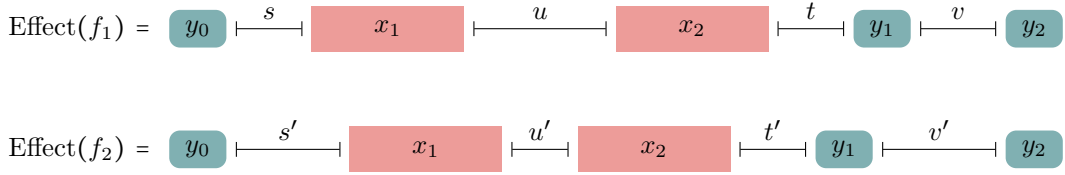
\[
\begin{tikzpicture}
\begin{scope}[yscale=0.2]
\draw (0,0) node [left] {$\eff{f_1} \,=$};
\draw (0,0) node (A) {};
\draw (1,0) node (B) {};
\draw (2,0) node (C) {};
\draw (4,0) node (D) {};
\draw (6,0) node (E) {};
\draw (8,0) node (F) {};
\draw (9,0) node (G) {};
\draw (10,0) node (H) {};
\draw (11,0) node (I) {};
\draw (12,0) node (J) {};
\draw [draw=none,fill=nicered,opacity=0.5] ([yshift=-1.5cm]C) rectangle ([yshift=1.5cm]D);
\draw [draw=none,fill=nicered,opacity=0.5] ([yshift=-1.5cm]E) rectangle ([yshift=1.5cm]F);
\draw [draw=none,fill=nicecyan,opacity=0.5,rounded corners] 
      ([yshift=-1.25cm]A.east) rectangle ([yshift=1.25cm]B.west);
\draw [draw=none,fill=nicecyan,opacity=0.5,rounded corners] 
      ([yshift=-1.25cm]G.east) rectangle ([yshift=1.25cm]H.west);
\draw [draw=none,fill=nicecyan,opacity=0.5,rounded corners] 
      ([yshift=-1.25cm]I.east) rectangle ([yshift=1.25cm]J.west);
\draw ($(A)!0.5!(B)$) node {$y_0$};
\draw (B.center) edge [|-|] node [above] {$s$} (C.west);
\draw ($(C)!0.5!(D)$) node {$x_1$};
\draw (D.east) edge [|-|] node [above] {$u$} (E.west);
\draw ($(G)!0.5!(H)$) node {$y_1$};
\draw (F.east) edge [|-|] node [above] {$t$} (G.center);
\draw ($(E)!0.5!(F)$) node {$x_2$};
\draw (H.center) edge [|-|] node [above] {$v$} (I.center);
\draw ($(I)!0.5!(J)$) node {$y_2$};
\end{scope}
\begin{scope}[yshift=-1.6cm,yscale=0.2]
\draw (0,0) node [left] {$\eff{f_2} \,=$};
\draw (0,0) node (A) {};
\draw (1,0) node (B) {};
\draw (2.5,0) node (C) {};
\draw (4.5,0) node (D) {};
\draw (5.5,0) node (E) {};
\draw (7.5,0) node (F) {};
\draw (8.5,0) node (G) {};
\draw (9.5,0) node (H) {};
\draw (11,0) node (I) {};
\draw (12,0) node (J) {};
\draw [draw=none,fill=nicered,opacity=0.5] ([yshift=-1.5cm]C) rectangle ([yshift=1.5cm]D);
\draw [draw=none,fill=nicered,opacity=0.5] ([yshift=-1.5cm]E) rectangle ([yshift=1.5cm]F);
\draw [draw=none,fill=nicecyan,opacity=0.5,rounded corners] 
      ([yshift=-1.25cm]A.east) rectangle ([yshift=1.25cm]B.west);
\draw [draw=none,fill=nicecyan,opacity=0.5,rounded corners] 
      ([yshift=-1.25cm]G.east) rectangle ([yshift=1.25cm]H.west);
\draw [draw=none,fill=nicecyan,opacity=0.5,rounded corners] 
      ([yshift=-1.25cm]I.east) rectangle ([yshift=1.25cm]J.west);
\draw ($(A)!0.5!(B)$) node {$y_0$};
\draw (B.center) edge [|-|] node [above] {$s'$} (C.west);
\draw ($(C)!0.5!(D)$) node {$x_1$};
\draw (D.east) edge [|-|] node [above] {$u'$} (E.west);
\draw ($(G)!0.5!(H)$) node {$y_1$};
\draw (F.east) edge [|-|] node [above] {$t'$} (G.center);
\draw ($(E)!0.5!(F)$) node {$x_2$};
\draw (H.center) edge [|-|] node [above] {$v'$} (I.center);
\draw ($(I)!0.5!(J)$) node {$y_2$};
\end{scope}
\end{tikzpicture}
\]

It is important to note that in any word $\eff{f}$ the 
variables from $X$ occur in the standard order $x_1<\dots<x_m$,
and similarly for the variables from $Y$.
In general, it may happen that, due to updates that 
concatenate registers together, the $x$'s and the $y$'s are 
not strictly interleaved one with the other (as an example, 
see $\eff{f_1}$ in the figure above). 
Here however, since $\t_1$ and $\t_2$ were assumed
to have the same flow, we know that the interleaving of the
$x$'s and the $y$'s is the same in $\eff{f_1}$ and $\eff{f_2}$.

Of course, since  $\out(\r\,\t_1\,\s) = \out(\r\,\t_2\,\s)$ the occurrences 
of the first and the last variables, $y_0$ and $y_m$,
are aligned exactly, as they represent the same extremal gaps.
For the remaining variables, which we generically denote $z_1,\dots,z_\ell$,
we proceed by splitting the equation $\eff{f_1}=\eff{f_2}$, devoid of the
extremal variables, into sub-equations that involve fewer variables, and 
reason by induction. Hereafter, $L=R$ denotes an equation over the 
variables $z_1,\dots,z_\ell$, 
that occur exactly once on each side 
of the equation and with the same order. Moreover, the factors of
$L$ and $R$ over $\G$ have length at most $c$, the capacity of the SST.

Suppose that $L = s\, z_1\, L'$ and $R= s'\, z_1\, R'$, 
with $s,s' \in \G^*$ and 
$L',R' \in \G^* z_{2} \, \G^* \dots \, \G^* z_\ell$.
Assume from now on that $z_1$ is a register (gaps are treated
symmetrically). 

If $s=s'$, then
the occurrences of $z_1$
perfectly aligned, and so $L'=R'$. We can then use induction.
Note that in this case there is no restriction on $z_1$. In particular,
the approximant $\App(z_1)$ for the valuation of $z_1$ induced by the
initial run $\r$ can well be $\G^*$. Moreover, any solution of the equation
$L=R$, where we replace the valuation for $z_1$ by an arbitrary word from 
its approximant $\App(z_1)$, is again a solution.
 
Otherwise, if $s\neq s'$, then either $s$ is a prefix of $s'$, or the other
way around. Suppose by symmetry that $s'=s\, w$ for some $w\in\G^+$,
hence $|w| \le c$.
We get the equation $w \, z_1 = z_1 \, w'$, where $w'$ is a conjugate of $w$, 
i.e.~$w=w_1\,w_2$ and $w'=w_2\,w_1$ for some $w_1,w_2\in\G^*$.
It follows that in every solution of $L=R$, 
the value of $z_1$ must range over the periodic language $w^* w_1$.
If we consider the register valuation $\nu = \val{\r\pos}$
induced by the initial run $\r$, then we have $\nu(z_1) \in w^* w_1$. 

Now, we assume without loss of generality that $w$ is primitive.
Since the length of $w$ is at most $c$ and since $\alpha\ge c$, 
we know that $w^* w_1$ is an $\alpha$-approximant.
Moreover, since $\nu(z_1) \in w^* w_1$ and $\App=\closure{(\val{\r\pos})}$, 
we get $\App(z_1) \subseteq w^* w_1$ (in particular, $\App(z_1)$ can be either
a singleton or the periodic language $w^* w_1$ itself). 
This implies that the equation $s\,z_1\,w' = s'\, z_1$ holds for any
word from $\App(z_1)$.

For the remaining variables, we observe that, up to any valuation
that satisfies $\eff{f_1}=\eff{f_2}$,  we get a new equation $w'\, L' = R'$ in a fewer 
number of variables. From there by applying induction we get the desired
claim.
\end{proof}

\medskip
Recall that in an SST that admits $\alpha$-approximants, 
states are of the form $(q,\App_X,\App_Y)$, 
and we have $(\closure{(\val{\r\pos})})|_X = \App_X$ 
(resp.~$(\closure{(\val{\pos\s})})|_Y=\App_Y$) for every initial run $\r$
that ends in $(q,\App_X,\App_Y)$ (resp.~for every final run $\s$
that starts in $(q,\App_X,\App_Y)$).
By pairing this with Lemma \ref{lem:equivalence-by-approximants}, 
we immediately obtain the following corollary:

\begin{restatable}{corollary}{CorEquivalenceByApproximants}\label{cor:equivalence-by-approximants}
Let $\T$ be a trimmed flow-normalized SST.
One can decide in polynomial time whether two given transitions 
$\t_1,\t_2$ of $\T$ with the same flow are equivalent. 
Moreover, if $\T$ has capacity $c$ and admits $\alpha$-approximants 
for some $\alpha\ge c$, then only two cases can happen:
\begin{enumerate}
\item either $\out(\r \, \t_1 \, \s) = \out(\r \, \t_2 \, \s)$ for every context $(\r,\s)$
      (so $\t_1,\t_2$ are equivalent),
\item or $\out(\r \, \t_1 \, \s) \neq \out(\r \, \t_2 \, \s)$ for \emph{every} context $(\r,\s)$
      (so $\t_1,\t_2$ are not equivalent).
\end{enumerate}
\end{restatable}

\medskip
Another important consequence is the following theorem, that normalizes 
finite-valued SST in order to bound the maximum number 
of transitions linking the same pair of states and consuming the 
same input letter. This number is called \emph{edge ambiguity} for short.

\begin{restatable}{theorem}{ThmEdgeAmbiguity}\label{thm:edge-ambiguity}
Let $\T$ be a $k$-valued, flow-normalized SST that has $m$ registers,
capacity $c$, and that admits $\alpha$-approximants, for some $\alpha\ge c$.
One can construct an equivalent SST $\T'$, with the same states and the 
same registers as $\T$, that has edge ambiguity at most 
$\edgeambiguity$.
\end{restatable}

\begin{proof}
By Corollary \ref{cor:equivalence-by-approximants}, 
$\T$ has at most $k$ pairwise non-equivalent transitions 
with the same input letter, the same source and target states,
and the same flow. Moreover equivalence of such transitions 
can be decided. 
We can then normalize $\T$ by removing in each equivalence class
all but one transitions with the same flow. Since $\T$ has at 
most $\normflows$ flows, the normalization results in an equivalent 
SST $\T'$ with edge ambiguity at most $\edgeambiguity$.
\end{proof}

The next section is devoted to prove a very similar result as above, but 
for all SST that can be obtained by shortcutting runs into single transitions, 
and that thus have arbitrary large capacity. This will be the main
technical ingredient for establishing the decidability of the equivalence 
problem for $k$-valued SST.

\section{Shortcut construction}\label{sec:shortcut}

Here we focus on a transformation of relations that absorbs the first
input letter when this is equal to a specific element, say $a\in\Sigma\setminus\{\dashv\}$.
Such a transformation maps any relation $R$ to the relation 
$R_a = \{(u,v) \mid (au,v)\in R\}$.
Observe that $R = R_\emptystr \,\cup\, \bigcup_{a\in\Sigma} R_a $,
where $R_\emptystr = R \:\cap\: (\{\dashv\}\times\Gamma^*)$.

It is easy to see that the class of relations realized by SST is
effectively closed under the transformation $R \mapsto R_a$. 
To prove this closure property, it is convenient to restrict,
without loss of generality, to SST with \emph{transient initial states}, 
namely, SST where no transition reaches an initial state.
Under this assumption, the closure property also preserves the 
state space (though some states may become useless), the set of 
registers, the property of being $k$-valued, as well as the $\alpha$-approximants,
if they are admitted by the original SST. 
However, the transformation does not preserves the capacity, which may increase.

\begin{restatable}{lemma}{LemShortcutClosure}\label{lem:shortcut-closure}
Given a flow-normalized SST $\T$ with transient initial states, 
and given a letter $a\in\Sigma\setminus\{\dashv\}$, one can construct 
an SST $\T_a$ with transient initial states such that
\begin{itemize}
\item $\dom{\T_a}=\{u\in\Sigma^* \mid au\in\dom{\T}\}$,
\item $\T_a$ on input $u$ produces the same outputs as $\T$ on input $au$.
\end{itemize}
Moreover, $\T_a$ has the same states and the same registers 
as $\T$; if $\T$ has capacity $c$, then $\T_a$ has
capacity $2c$; 
if $\T$ admits $\alpha$-approximants, then so does $\T_a$ 
(via the same annotation).
\end{restatable}

\begin{proof}
The construction is rather straightforward and boils down to shortcutting
the first $a$-labeled transition in every successful run. 
Given $\Tfulln$, we define $\Tfulln[T_a]<Q><U'><I><E'><F>$,
where $U'=U \,\cup\, U\circ U$, $\circ$ denotes the functional composition,
and $E'$ contains the following transitions:
\begin{itemize}
\item $q_0\trans{b / f\circ f'} q'$, if $E$ contains some transitions 
      $q_0\trans{a / f} q \trans{b / f'} q'$, with $q_0$ initial state and $b\in\Sigma$,
\item $q\trans{b / f} q'$, if $E$ contains a transition $q\trans{b/f} q'$, with 
      $q$ is not initial and $b\in\Sigma$.
\end{itemize}
Note that, thanks to the assumption that every input of an SST ends with the special marker $\dashv$,
there is no final state in $\T$ that is a successor of an initial state along an 
$a$-labeled transition (unless of course $a=\dashv$, which we assumed to be not the case). 
This essentially means that any initial $a$-labeled transition
can be absorbed into the subsequent transitions, as precisely done in the above construction.

It is routine to check that $\T_a$ satisfies the desired claims.
Here we only show that $\T_a$ admits the same $\alpha$-approximants as $\T$.
This follows from the fact that every successful run $\rho$ of $\T_a$
of the form
\[
  \rho: ~ (q_0,\App_{0,X},\App_{0,Y}) 
          \trans{b/f}[\T_a] (q_1,\App_{1,X},\App_{1,Y}) 
          \trans{u/g}[\T_a] (q_2,\App_{2,X},\App_{2,Y})
\]
can be turned to a successful run $\rho'$ of $\T$ of the form
\[
  \rho': ~ (q_0,\App_{0,X}) 
           \trans{a/f_a}[\T] (q',\App'_X,\App'_Y) 
           \trans{b/f_b}[\T] (q_1,\App_{1,X},\App_{1,Y}) 
           \trans{u/g}[\T] (q_2,\App_{2,X},\App_{2,Y})
\]
where $f_a\circ f_b=f$.
In particular, we have
$\val{\rho,0}=\val{\rho',0}$ and 
$\val{\rho,i}=\val{\rho',i+1}$ for all positions $i>0$,
and hence both $\T$ and $\T_a$ satisfy Equation \eqref{eq:approximants}.
\end{proof}

The above construction can be applied inductively to compute 
an SST for any \emph{left quotient} $R_u=\{(v,w) \mid (u\,v,w) \in R\}$ of $R$. 
For $u=a_1\dots a_n\in(\Sigma\setminus\{\dashv\})^*$,
we let $\T_u = (\dots(\T_{a_1})_{a_2}\dots)_{a_n}$.

\medskip
The last and most technical step consists in proving that edge 
ambiguity can be uniformly bounded in every SST $\T_u$, provided 
that the initial SST $\T$ is finite-valued, flow-normalized, and 
admits $\alpha$-approximants for a large enough $\alpha$.
More precisely, we aim at establishing that, for $\alpha$ much larger
than the capacity of $\T$, the notion of $\alpha$-approximant,
besides satisfying Equation \eqref{eq:approximants}, 
also satisfies the following property:
\begin{equation}\label{eq:tight-approximants}
\begin{array}{llr@{~~}c@{~~}l}
  \forall \beta\ge\alpha
  &
  \exists \r ~:
  &\qquad
  \closure[\beta]{(\val{\r\pos})} &=& \closure{(\val{\r\pos})}
  \\[1ex]
  \forall \beta\ge\alpha
  &
  \exists \s ~:
  &\qquad
  \closure[\beta]{(\val{\pos\s})} &=& \closure{(\val{\pos\s})}.
\end{array}
\end{equation}
In this case we say that the $\alpha$-approximants are \emph{tight}.

Intuitively, the above property can be explained as follows.
When considering an SST with states annotated with $\alpha$-approximants,
it may happen that for some larger parameter $\beta$ some initial 
runs induce register valuations at a state $(q,\App_X,\App_Y)$ whose
$\beta$-approximants are strictly included in $\App_X$
(e.g.~possibly entailing new periodicities).
These runs should be thought of as exceptional cases,
and there is a way of pumping them so as to restore
the equality between $\App_X$ and the induced $\beta$-approximant. 

Let us first see how tight approximants are used.
The theorem below assumes that there is an SST $\T'$ with tight approximants
(later we will show how to compute such an SST), and bounds the edge 
ambiguity of the SST $\T'_u$ that realizes a left quotient of $T'$. 

\begin{restatable}{theorem}{ThmEdgeAmbiguityShortcut}\label{thm:edge-ambiguity-shortcut}
Let $\T'$ be a $k$-valued, flow-normalized SST realizing $R$, 
with transient initial states and tight $\alpha$-approximants.
For every $u\in\Sigma^*$, one can construct an SST $\T'_u$ 
realizing $R_u$, with the same states and the same registers
as $\T'$, and with edge ambiguity at most $\edgeambiguity$.
\end{restatable}

\begin{proof}
The crux is to show that the SST $\T'_u$ obtained from
Lemma \ref{lem:shortcut-closure} has at most $k$ pairwise 
non-equivalent transitions with the same flow (for any given
source/target state and label). 
Once this is proven, one can 
proceed as in the proof of Theorem \ref{thm:edge-ambiguity}, 
by removing 
all but one transition with the same flow in each 
equivalence class. 
By way of contradiction, assume that
$\T'_u$ has $k+1$ pairwise non-equivalent transitions
$\t_1,\dots,\t_{k+1}$ with the same flow.
Since $\T'$ admits tight $\alpha$-approximants, 
by Lemma \ref{lem:shortcut-closure} we know that 
the source and target state, respectively, of the previous transitions
are annotated with tight $\alpha$-approximants,
say $(\App_X,\App_Y)$ and $(\App'_X,\App'_Y)$, respectively.

We begin by applying the first claim of Lemma \ref{lem:equivalence-by-approximants},
implying that the equation $\eff{f_i} = \eff{f_j}$ is violated 
for some valuation $\nu\in\App_X\uplus\App'_Y$.
Then, we let $\beta = \max(\alpha,|u|c)$ and use 
Equations \eqref{eq:approximants} and \eqref{eq:tight-approximants} 
to get a context $(\r,\s)$ such that
$\closure[\beta]{(\val{\r\pos})} = \App_X$ and
$\closure[\beta]{(\val{\pos\s})} = \App'_Y$.
Finally, knowing that $\beta$ is at least the capacity of $\T'_u$,
we apply the second claim of Lemma \ref{lem:equivalence-by-approximants}
to get $\out(\r\,\t_i\,\s) \neq \out(\r\,\t_j\,\s)$,
thus witnessing non-equivalence of all pairs of transitions 
$\t_i,\t_j$ at the same time. This contradicts 
the assumption that $\T'_u$ (and hence $\T'$) is $k$-valued.
\end{proof}

Now, let $\T$ be a flow-normalized SST with $m$ registers, capacity $c$, and 
trimmed state space $Q$. Below, we show how to compute, with the help
of Proposition \ref{prop:approximants}, an SST $\T'$ equivalent to $\T$ 
that admits tight approximants.
For simplicity, we will mostly focus on register valuations induced by initial 
runs, even though similar results can be also stated for gap valuations induced 
by final runs.
We begin by giving a few technical results based on pumping arguments. 
We say that register $x$ is \emph{productive} along $\r$ if the 
update induced by $\r$ maps $x$ to a word that contains at least 
one letter from $\Gamma$. We also recall that a loop of a run
needs to induce a flow-idempotent update.

\begin{restatable}{lemma}{LemPumping}\label{lem:pumping}
If $\r = \r_1 \, \g\, \r_2$ is an initial run of $\T$, 
with $\gamma$ loop, then for every $n>0$ the pumped run 
$\r^{(n)} = \r_1\,\g^n\,\r_2$ induces valuations 
$\val{\rho^{(n)}\pos}$ mapping any register $x$ to a word of the form
$u_0 \, v_1^{n-1} \, u_1 \, \dots \, v_{2m}^{n-1} \, u_{2m}$, where 
$u_0,\dots,u_{2m},v_1,\dots,v_{2m}\in\Gamma^*$ depend on $\r$ and $x$, 
but not on $n$.
Moreover, we have $v_i\neq\emptystr$ for some $i$ if there is a register 
$x'$ that is productive along $\g$ and that flows into $x$ along 
$\r_2$.
\end{restatable}

\begin{proof}
We begin by observing a useful property of updates induced by loops:

\begin{claim}\label{claim:idempotent-flow}
If $\g$ is a loop,
then the list of registers $x_1,\dots,x_m$ 
can be partitioned into intervals $X_1<\dots<X_k$
such that for every $1\le i\le k$, there is $x'\in X_i$
so that every $x\in X_i$ flows into $x'$ along $\g$.
\end{claim}

\begin{claimproof}[Proof of claim]
It suffices to verify that idempotent flows always have shapes similar
to the flow below, 
where $X_1 = \{x_1,x_2,x_3\}$ and $X_2 = \{x_4,x_5,\dots,x_7\}$:
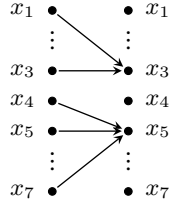
\[
\begin{tikzpicture}
\begin{scope}[yscale=0.4]
\draw (-0.1,0) node [left] {\small $x_1$};
\draw (-0.1,-2) node [left] {\small $x_3$};
\draw (-0.1,-3) node [left] {\small $x_4$};
\draw (-0.1,-4) node [left] {\small $x_5$};
\draw (-0.1,-6) node [left] {\small $x_7$};
\draw (1.1,0) node [right] {\small $x_1$};
\draw (1.1,-2) node [right] {\small $x_3$};
\draw (1.1,-3) node [right] {\small $x_4$};
\draw (1.1,-4) node [right] {\small $x_5$};
\draw (1.1,-6) node [right] {\small $x_7$};
\draw (0,0) node [dot] (1) {};
\draw (0,-1) node (2) {$\vdots$};
\draw (0,-2) node [dot] (3) {};
\draw (0,-3) node [dot] (4) {};
\draw (0,-4) node [dot] (5) {};
\draw (0,-5) node (6) {$\vdots$};
\draw (0,-6) node [dot] (7) {};
\draw (1,0) node [dot] (1') {};
\draw (1,-1) node (2') {$\vdots$};
\draw (1,-2) node [dot] (3') {};
\draw (1,-3) node [dot] (4') {};
\draw (1,-4) node [dot] (5') {};
\draw (1,-5) node (6') {$\vdots$};
\draw (1,-6) node [dot] (7') {};
\draw [arrow] (1) to (3');
\draw [arrow] (3) to (3');
\draw [arrow] (4) to (5');
\draw [arrow] (5) to (5');
\draw [arrow] (7) to (5');
\end{scope}
\end{tikzpicture}
\qedhere
\]
\end{claimproof}

The next claim allows to simplify
the statement of the lemma by assuming that $\r_2$ 
is empty. Its proof is straightforward, since $\T$ is non-erasing and
non-permuting. 

\begin{claim}\label{claim:flows-into}
If $\rho=\rho_1\,\rho_2$ is an initial run 
and $x_1,\dots,x_k$ flow into $x$ along $\rho_2$, 
then $\val{\rho\pos}(x)$ contains the factors 
$\val{\rho_1\pos}(x_1)$, \dots, $\val{\rho_1\pos}(x_k)$ 
in this precise order, possibly interleaved by 
other words that depend only on $\r_2$.
Moreover, if any of the $x_i$'s is productive 
along $\rho_1$, then so is $x$ along $\rho$.
\end{claim}

It now remains to prove that:

\begin{claim}\label{claim:pumping-idempotent}
If $\r = \r_1 \, \g$ is an initial run, 
with $\g$ loop, then for every $n>0$
the pumped run $\r^{(n)} = \r_1\,\g^n$ 
induces valuations $\val{\rho^{(n)}\pos}$ 
mapping any register $x$ to a word of the form
$v_1^{n-1} \, u \, v_2^{n-1}$,
where $u,v_1,v_2\in\Gamma^*$ depend 
on $\r$ and $x$, but not on $n$.
Moreover, $v_1$ or $v_2$ is non-empty if $x$ 
is productive along $\g$.
\end{claim}

We use  claim~\ref{claim:idempotent-flow} and 
for simplicity we work on the example provided there.
Let us assume on the example that the updates are as follows
(recall that $\Tt$ is non-permuting):
\begin{itemize}
\item $f(x_3) = t_1 \, x_1 \, t_2 \, x_2 \, t_3 \, x_3 \, t_4$,
\item $f(x_5) = t_5 \, x_4 \, t_6 \, x_5 \, t_7 \, x_6 \, t_8 \,  x_7 \, t_9$,
\item $f(x_j) = t'_j$, for all remaining $j \neq 3,5$.
\end{itemize}
Let $\nu=\val{\rho_1\pos}$ be the register valuation induced by the 
prefix $\rho_1$. The valuation $\val{\rho^{(n)}\pos} = \nu\circ f^n$ 
maps e.g.
\begin{itemize}
\item $x_3$ to 
      $\big(t_1 \, t'_1 \, t_2 \, t'_2\, t_3\big)^{n-1} ~ 
       \big(t_1 \, \nu(x_1) \, t_2 \, \nu(x_2)\, t_3\big) ~ \nu(x_3) ~~ 
       \big(t_4\big)^n$,
\item $x_5$ to 
      $\big(t_5 \, t'_4 \, t_6\big)^{n-1} ~ \big(t_5 \, \nu(x_4) \, t_6\big) ~ 
       \nu(x_5) ~ 
       \big(t_7 \, \nu(x_6) \, t_8 \, \nu(x_7) \, t_9\big) ~~
       \big(t_7 \, t'_6 \, t_8 \, t'_7 \, t_9\big)^{n-1}$.
\end{itemize}
The claim is satisfied e.g.~for $x=x_3$ by setting
$v_1 = t_1 \, t'_1 \, t_2 \, t'_2\, t_3$, 
$u= \big(t_1 \, \nu(x_1) \, t_2 \, \nu(x_2)\, t_3\big) ~ \nu(x_3) ~ t_4$, 
and
$v_2 = t_4$.
\end{proof}

Given a tuple of pairwise disjoint loops 
$\bar\g = \g_1,\dots,\g_\ell$ in a run $\r$, 
we write $\r' \yields[\bar\g] \r$ when $\r'$ is obtained from $\r$ by 
{\sl simultaneously} pumping $n$ times every loop $\g_i$, for some $n>0$.
When using this notation, we often omit the subscript $\bar\g$; 
in this case we tacitly assume that $\bar\g$ is {\sl uniquely determined} from $\r$. 
In this way, when writing, for instance, $\r',\r''\yields\r$, 
we will know that $\r',\r''$ are obtained by pumping the same loops of $\r$.
We also say that a property on runs holds \emph{for all but finitely many $\r'\yields\r$}
if it holds on runs $\r'$ that are obtained from $\r$ by pumping $n$ times the loops
in a fixed tuple $\bar\g$, for all $n>n_0$ and for a sufficiently large $n_0$.

\begin{restatable}{lemma}{LemLengthAndPeriod}\label{lem:length-and-period}
Let $\r$ be an initial run and $x$ a register.
If $\val{\r\pos}(x)$ has length (resp.~period) larger than $\alpha=\bound$, 
then for every $\beta\ge\alpha$ and for all but finitely many $\r' \yields \r$, 
$\val{\r'\pos}(x)$ has length (resp.~period) larger than $\beta$.
\end{restatable}

\begin{proof}
The first step consists in identifying the appropriate
loops $\g_1,\dots,\g_\ell$ inside the initial run $\r$. 
More precisely, we need to factorize $\r$ as
\[
  \r ~=~ \r_0 ~ \g_1 ~ \r_1 ~ \dots ~ \g_\ell ~ \r_\ell.
\]
where $\g_1,\dots,\g_\ell$ are loops, in such a way 
that every register with large enough induced valuation 
is productive along at least one loop.
In fact, for technical reasons related to periodicity,
we need to also guarantee that the selected loops only 
contribute for a bounded portion to the valuation of a register,
precisely, with at most $\boundnom$ letters.

For every register $z$ and every position $i$ of $\r$, 
let $X_{i,z}$ be the set of registers that flow into $z$ 
along the suffix of $\r$ that starts at position $i$.
Further let $N_{i,z} = \sum_{x\in X_{i,z}} |\val{\r,i}(x)|$,
and let $D_{i,j,z} = N_{j,z} - N_{i,z}$ for all $i\le j$.
To find a productive loop for $z$
between positions $i\le j$, it suffices to have a 
large enough value $D_{i,j,z}$:

\begin{claim}\label{claim:bound}
If $D_{i,j,z} > \boundnom$, then $\r$ contains a
loop $\g$ between positions $i$ and $j$, and 
there is a register $x$ that is productive 
along $\g$ and that flows into $z$ 
along the suffix of $\r$ that follows $\g$.
\end{claim}

\begin{claimproof}[Proof of claim]
Let $\s$ be the factor of $\r$ between positions $i$ and $j$.
Since $\T$ is copyless with capacity $c$ and $D_{i,j,z} > \boundnom$, 
there are $N>\boundnomc$ transitions between $i$ and $j$
along which some register in $X_{k,z}$ is productive.
Among these transitions, there are $n>\boundnomcQ$ 
that start with the same source state, say $q$.
Let $i_1 < \dots < i_n$ be the positions where the
latter transitions start.

Next, consider the flows $F_j$ of the updates induced between 
positions $i_j$ and $i_{j+1}$, for all $j=1,\dots,n$.
Recall that flows are naturally equipped with an associative 
product, forming a monoid $M$ of size at most $m^2$.
By the Factorization Forest theorem \cite{Simon90,colcombet07simon,kuf08mfcs}, 
there is a factorization tree for the sequence $F_1\dots F_n$ 
that has height at most $3|M|$ and such that every inner node 
with more than two successors has all children labeled by the 
same idempotent flow. 

Since $n>\boundnomcQ\ge \boundnomcQ[|M|]$, 
there is at least one idempotent flow $F_j$.
This proves that $\r$ contains a loop $\g$ between 
positions $i$ and $j$.
Moreover, there is a register $x$ that is productive along $\g$ 
and that flows into $z$ along the suffix of $\r$ that follows $\g$.
\end{claimproof}

We construct the desired factorization of $\r$ by induction
as follows.
We maintain a position $i$ in $\r$, representing the endpoint
of the processed prefix of $\r$, and a set $Z$ of registers for 
which we still need to find corresponding productive loops. 
The position $i$ is
initialized to $0$, and the set $Z$ to the set of registers
$z$ such that $|\val{\r\pos}(z)| > \bound$.
We then look at the first position $j>i$ such that 
$D_{i,j,z} > \boundnom$, for some $z\in Z$
(the construction terminates as soon as $Z$ becomes empty).
By Claim \ref{claim:bound}, we know that the factor of $\r$
between positions $i$ and $j$ contains a loop $\g$,
and there is a register $x$ that is productive 
along $\g$ and flows into $z$ along the suffix that follows $\g$. 
Moreover, thanks to the above eager strategy, the number of output 
letters that appear inside $g(x)$, where $g$ is the update induced 
by $\g$, is at most $\boundnom$.
We can thus declare $\g$ to be one of the loops 
of our factorization, and accordingly set $i$ to $j$ 
and remove $z$ from $Z$. 
Note that the following invariant is preserved:
for all $z\in Z$, $|\val{\r\pos}(z)| > \bound - D_{0,i,z}$.
Because at each iteration the value of $D_{0,i,z}$ increases
by at most $\boundnom$, and because at most $m$ iterations
are possible, this shows that the construction can carried
over correctly.

\medskip
We are now ready to prove the lemma. For the property
concerning the lengths of the register valuations, 
suppose that $|\val{\rho\pos}(x)| > \alpha = \bound$.
By the previous constructions, there is a loop $\g_i$
with a productive register $x'$ that flows into $x$ 
along the suffix $\r_i ~ \g_{i+1} ~ \dots ~ \g_\ell ~ \r_\ell$.
By Lemma \ref{lem:pumping}, the valuations induced 
at the end of the pumped runs 
\[
  \rho^{(n)} ~=~ 
  \r_0 \, \g_1 \, \r_1 \, \dots \, \pmb{\g_i^n} \, \r_i \, \dots \, \g_\ell \, \r_\ell
\]
map $x$ to arbitrarily long words. 
Moreover, the same can be 
said of the lengths of the valuations of $x$ that are
induced by runs obtained by pumping {\sl simultaneously},
and by the same amount $n$, all loops $\g_1,\dots,\g_\ell$.
This proves that, for every $\beta\ge\alpha$ and for all
but finitely many $\r'\yields \r$,
$|\val{\r'\pos}(x)| > \beta$.

We can use a similar argument to prove the property concerning
the periods. Suppose that $\val{\rho\pos}(x)$ has period 
$p > \alpha = \bound$. In particular, $|\val{\rho\pos}(x)|>\alpha$. 
As before, there is a loop $\g_i$ with a productive register $x'$ 
that flows into $x$ along the suffix 
$\r_i ~ \g_{i+1} ~ \dots ~ \g_\ell ~ \r_\ell$.
Moreover, by the previous constructions we know that the effect 
of the loop $\g_i$ on the final valuation of $x$ is to add at 
most $\boundnom$ letters.
Let us consider runs that are obtained by pumping simultaneously
all loops $\g_1,\dots,\g_\ell$ inside $\rho$:
\[
  \rho^{(n)} ~=~ 
  \r_0 \, \pmb{\g_1^n} \, \r_1 \, \dots \, \pmb{\g_i^n} \, \r_i \, \dots \, \pmb{\g_\ell^n} \, \r_\ell.
\] 
By Lemma \ref{lem:pumping} 
(plus Claim \ref{claim:flows-into}), 
the valuations induced at the end of the pumped runs $\rho^{(n)}$
map $x$ to words of the form
\[
  \val{\hat\rho^{(n)}\pos}(x) ~=~
  u_0 \, \pmb{v_1^{n-1}} \, u_1 
  \, \dots \, 
  u_{t-1} \, \pmb{v_t^{n-1}} \, u_t.
\]
for some $t \le 2m\ell$, where $u_0,\dots,u_t,v_1,\dots,v_t\in\G^*$ depend 
only on $\rho$ and $\bar\g$, $|v_i|\le\alpha$ for all $i\le t$, and
$|v_i|>0$ for some $i\le t$.
In particular, the above words contain arbitrarily long repetitions 
of non-empty words.

Now, let $p_n$ be the period of $\val{\rho^{(n)}\pos}(x)$, for all $n>0$. 
Recall that $p_1 = p > \alpha$. 
We aim at showing that the periods $p_n$ get arbitrarily large.
Suppose, by way of contradiction, that $p_n$ is uniformly bounded 
for all $n>0$. 
Then $p_n$ must be a constant, say $p_n=p'$, for infinitely many $n$. 
We also recall from the previous arguments that 
$\val{\rho^{(n)}\pos}(x)$ has arbitrarily long repetitions
of words of length $r_1=|v_1|$, \dots, $r_k=|v_k|$, with
all $r_j\le\alpha$ and at least one $r_j>0$.
By Fine-Wilf's theorem, 
this
implies that the period of $\val{\rho^{(n)}\pos}(x)$,
for infinitely many $n$, is 
\[
  p'' ~=~ \gcd\{p',r_i\}_{r_i>0} ~<~ \a ~<~ p.
\]
We can transfer this property to the original word $\val{\rho\pos}(x)$,
by observing that $\val{\rho\pos}(x)$ can be obtained from any of the 
previous words $\val{\rho^{(n)}\pos}(x)$ by removing some occurrences 
of factors of lengths $r_1,\dots,r_k$.
As those lengths are multiples of the period $p''$, 
the latter operation does not change the period of the entire word.
Hence, $\val{\rho\pos}(x)$ must also have period $p'' < p$, which is 
however a contradiction.

This proves that $p_n$ gets arbitrarily large for $n>0$.
In particular, for every $\beta\ge\alpha$ and for all but finitely many
runs $\rho' \yields \rho$, the word $\val{\r'\pos}(x)$ has period larger than $\beta$.
\end{proof}

Using the previous lemmas and the fact that the type of quantification 
``for all but finitely many runs'' commutes with conjunctions (e.g.~those
used to enforce properties on each register $x\in X$), we obtain that 
$\alpha$-approximants are tight for sufficiently large $\alpha$:

\begin{restatable}{proposition}{PropTightApproximants}\label{prop:tight-approximants}
Let $\T'$ be the SST admitting $\alpha$-approximants that is 
obtained from $\T$ using Proposition \ref{prop:approximants},
for any $\alpha \ge \bound$, where $Q$ is the set of states of $T$.
The $\alpha$-approximants of $\T'$ are tight.
\end{restatable}

\begin{proof}
As usual, by symmetry we can focus only on register valuations induced by initial runs.
We fix, once and for all, two parameters $\alpha,\beta$, with $\alpha\ge\bound$ and
$\beta\ge\alpha$.
For the sake of readability, we also introduce the shorthands
$\App_\r = \closure{(\val{\r\pos})}$ and $\Appbis_\r = \closure[\beta]{(\val{\r\pos})}$,
for any initial run $\r$ of $\T'$.
Since $\beta\ge\alpha$, we have $\Appbis_\r(x) \subseteq \App_\r(x)$.
We need to prove that there is an initial run $\r'$ of $\T'$ 
such that, for all registers $x$, $\Appbis_{\r'}(x) = \App_{\r'}(x)$.

We will in fact prove a slightly stronger claim, that is:
for all initial runs $\r$ or $\T'$, for all but finitely many
runs $\r'\yields\r$, and for all registers $x$, $\Appbis_{\r'}(x) = \App_{\r'}(x)$.
Towards this we analyse the possible cases when $\Appbis_\r(x)$ could be
strictly contained in $\App_\r(x)$, for any initial run $\r'$.
By definition of $\beta$-approximant, this could only happen when
$\Appbis_\r(x)$ contains only one word,  
or when it is a language of the form $u^* v$. 
In the former case we say for short that $\Appbis_{\r'}(x)$ is a \emph{singleton}; 
in the latter case we say that $\Appbis_{\r'}$ is a \emph{periodic language}.
We then prove that, for every register $x$ and every initial run 
$\r$ of $\T'$:
\begin{itemize}
	\item if $\Appbis_\r(x)$ is a singleton strictly included in $\App_\r(x)$,
	      then, {\sl for all but finitely many} runs $\r' \yields \r$, 
	      $\Appbis_{\r'}(x)$ is not a singleton; 
	\item if $\Appbis_\r(x)$ is a periodic language strictly included in $\App_\r(x)$,
	      then, {\sl for all but finitely many} runs $\r' \yields \r$,
	      $\Appbis_{\r'}(x)$ is not a periodic language (and thus neither a singleton).
\end{itemize}
Note that the quantification ``for all but finitely many'', 
like universal quantification, commutes with the conjunction over the registers $x$.
Therefore, the above two properties, paired with the previous arguments, 
suffice to prove the desired claim.

Now, fix a register $x\in X$ and an initial run 
$\r$ of $\T'$, and suppose that $\Appbis_\r(x)$ is a singleton
or a periodic language strictly contained in $\App_\r(x)$.

If $\Appbis_\r(x)$ is a singleton, say $\Appbis_\r(x)=\{u\}$,
then, since $\closure{u}=\App_\r(x) \supsetneq \{u\}$, we know that $|u| > \alpha$.
Recall that $\T'$ is a covering of $\T$, and in particular that the (initial) runs 
of $\T'$ are bijectively related to the (initial) runs of $\T$.
Let $\tilde\r$ be the initial run of $\T$ that corresponds to $\r$.
By Lemma \ref{lem:length-and-period}, we get that, 
for all but finitely many runs $\tilde\r' \yields \tilde\r$ of $\T$, 
$\val{\tilde\r'\pos}(x)$ has length even larger than $\beta$.
By exploiting again the bijection between runs of $\T$ and runs of $\T'$, we get
that, for all but finitely many runs $\r' \yields \r$ of $\T'$, 
the word $\val{\r'\pos}(x)$ has length larger than $\beta$,
and hence $\Appbis_{\r'}(x)$ cannot be a singleton.

If $\Appbis$ is a periodic language of the form $u^* v$,
then we get $|u| > \alpha$, and hence the period of $\val{\r'\pos}(x)$ 
is larger than $\alpha$. 
Using the correspondence between runs of $\T$ and runs of $\T'$ 
and exploiting Lemma \ref{lem:length-and-period}, exactly as we did before, 
we get that, for all but finitely many runs $\r' \yields \r$ of $\T'$, 
$\val{\r'\pos}(x)$ has period larger than $\beta$, and hence
$\Appbis_{\r'}$ is not a periodic language.
\end{proof}

\section{Equivalence algorithm}

The equivalence algorithm for $k$-valued SST follows a classical approach of
Culik and Karhumäki \cite{ck86} that is based on so-called \emph{test
  sets}. A test set for two SST $\T_1,\T_2$ over input alphabet $\S$
is a set $F \subseteq \S^*$ such that $\T_1,\T_2$ are equivalent if
and only if they are equivalent over $F$. The main contribution of
\cite{ck86} is to show that {\sl finite} test sets exist and be computed
effectively for 
$k$-valued one-way transducers.
The key ingredient of their proof is
to show the existence of a test set that works for \emph{all} transducers with
fixed number of states. An essential observation
is that for $k$-valued one-way, or even two-way, transducers one can assume
that the edge ambiguity is at most $k$. The reason for this is simply that the
output is generated sequentially. For SST the situation 
is far more complex because the output is generated piecewise. 
The purpose of the normalizations performed in
Section~\ref{sec:normalizations} was precisely to 
restore the property of bounded edge ambiguity. 

In a nutshell, the existence of a test set for transducers 
is a consequence of Ehrenfeucht's conjecture, whereas the 
effectiveness is based on the resolution of word equations 
due to Makanin~(see e.g.~the survey~\cite{die02makanin}).

Ehrenfeucht's conjecture was originally stated as a conjecture about
formal languages: for every language $L \subseteq \S^*$, there is
a finite subset $F \subseteq L$ such that for all morphisms
$f,g : \Sigma^* \to \Delta^*$,
$f(w)=g(w)$ for every $w \in L$ if and only if $f(w)=g(w)$ for every $w \in F$. 
Such a set $F$ is called a \emph{test set} for $L$. 

There is an
equivalent formulation of Ehrenfeucht's conjecture in terms of 
a compactness property of word equations~\cite{kar84EF}. 
Let $\S$ and $\Omega$ be two alphabets,
where the elements in $\Omega$ are called unknowns.  
A word equation is a pair $(u,v) \in \Omega^* \times \Omega^*$, 
and a solution is a morphism $\s: \Omega^* \to \Sigma^*$ such that $\s(u)=\s(v)$.  
Ehrenfeucht's conjecture is equivalent to saying that any system of 
equations over a finite set $\Omega$ of unknowns has a finite, 
equivalent subsystem, where equivalence means
that the solution sets are the same. The latter compactness property 
was proved in~\cite{AL85,Guba86} by encoding words by polynomials and
using Hilbert's basis theorem.

\bigskip
In view of Propositions \ref{prop:flow-normalization}, \ref{prop:approximants},
and \ref{prop:tight-approximants}, we can restrict without loss of generality
to SST that are flow-normalized and that admit tight approximants. 
Hereafter, we shall tacitly assume that all transducers are of this form.
Given some integers $k$, $n$, $m$, and $e$, let $\Cc_k(n,m,e)$ 
be the class of $k$-valued SST with at most $n$ states, 
$m$ registers, and edge-ambiguity at most $e$. 
Note that if $\T$ is $k$-valued, then by Theorem~\ref{thm:edge-ambiguity} 
it belongs to $\Cc_k(n,m,e)$, where $n,m$ are the number of states and 
registers of $\T$ and $e=\edgeambiguity$. 
Similarly, by Lemma \ref{lem:shortcut-closure} and 
Theorem \ref{thm:edge-ambiguity-shortcut},
every left quotient $\T_u$ also belongs to $\Cc_k(n,m,e)$.

Now, let us fix $k,n,m,e$ and consider an arbitrary SST $\T$ 
from $\Cc_k(n,m,e)$. 
Following \cite{ck86} we first build  an
abstraction of $\T$ by replacing each maximal factor from $\G^*$
occurring in some update function of $\T$, by a distinct unknown
from $\Omega$. The SST $\D(\T)$ obtained in this way is called a 
\emph{schema}; its outputs are words over $\Omega$. 
Note that the assumption of bounded edge ambiguity is essential 
here to get a uniform bound on the number of unknowns required 
for a schema. Clearly, there are only finitely many schemas of 
SST in $\Cc_k(n,m,e)$. 
We denote by $\phi_{\T}:\Omega \rightharpoonup \G^*$ the partial 
mapping (concretization) that associates with each unknown
the corresponding word from $\G^*$ as specified by the updates 
of $\T$.

We can rephrase the equivalence $\T_1 \equiv \T_2$ of two arbitrary 
 SST from $\Cc_k(n,m,e)$ 
as an infinite ``system'' of word equations
\footnote{Formally, $\Ss$ depends on $n$ and $k$, 
          but for simplicity we leave out the indices.} 
$\Ss= \bigwedge_{u \in\S^*} \bigvee_{\pi} S_{\pi}$ 
over set of unknowns $\Omega \uplus \Omega'$.
The unknowns from $\Omega$ are used for the schema
$\D(\T_1)$, whereas those from $\Omega'$ are used for $\D(\T_2)$;
in particular, $\phi_{\T_1} : \Omega \to \G^*$ and $\phi_{\T_2}: \Omega' \to \G^*$. 
The disjunctions in $\Ss$ are finite, with $\pi$ ranging over the 
possible schemas $\Delta_1,\Delta_2$ (for $\T_1$ and $\T_2$, respectively) 
and the possible partitions of the set of runs of $\Delta_1$ and 
$\Delta_2$ over the input $u$, into at most $k$ groups (one for each possible output). 
Finally, $S_{\pi}$ is a (finite) system of word equations, 
stating the equality of the words from $\Omega^* \cup {\Omega'}^*$ 
that belong to the same group according to $\pi$.

The following lemma was stated in \cite{ck86} for $k$-valued one-way
transducers, but it holds as well for two-way transducers and for SST
(even copyful, with a proper definition for $\Cc_k(n,m,e)$):

\begin{restatable}{lemma}{LemEquations}\label{lem:equations}
Given two SST $\T_1,\T_2$ from $\Cc_k(n,m,e)$, the system 
$\Ss = \bigwedge_{u \in\S^*} \bigvee_{\pi} S_{\pi}$ has
$\phi_{\T_1} \uplus \phi_{\T_2}$ as solution if and only if $\T_1 \equiv \T_2$.
\end{restatable}

As shown in \cite{ck86}, the Ehrenfeucht conjecture can be used to show 
that any infinite system $\Ss$ as in Lemma~\ref{lem:equations} 
is equivalent to some \emph{finite} sub-system  
$\Ss_N= \bigwedge_{u \in\S^{\le N}} \bigvee_{\pi} S_{\pi}$. 
This gives:

\begin{restatable}{lemma}{LemTest}\label{lem:test}
Given $n,m,e \in \Nat$, there is $N \in \Nat$ such that
$\S^{\le N}$ is a test set for every pair of SST 
$\T_1,\T_2$ from $\Cc_k(n,m,e)$. 
\end{restatable}

Using Theorem~\ref{thm:edge-ambiguity} and Lemma~\ref{lem:test}  we can derive 
immediately the existence of a finite test set for any two $k$-valued SST. 
The last question is how to compute such a test set effectively. For this we
will use the shortcut construction provided in
Section~\ref{sec:shortcut}. 

\begin{restatable}{lemma}{LemEffective}\label{lem:effective}
Assume that the formulas $\Ss_N$ and $\Ss_{N+1}$ are
equivalent, i.e., they have the same solutions. 
Then $\Sigma^{\le N}$ is a test set for any 
 pair of SST from $\Cc_k(n,m,e)$.
\end{restatable}

\begin{proof}
Let $\T_1 \equiv_r \T_2$ denote equivalence of $\T_1$ and $\T_2$ 
relativized to $\S^{\le r}$.
The goal is to prove that $\S^{\le N}$ is a test set, 
namely, for all $r>N$ and all $\T_1,\T_2\in\Cc_k(n,m,e)$,
$\T_1 \equiv_r \T_2$ holds if and only if $\T_1 \equiv_N \T_2$. 
Clearly, for any $r \ge 0$, $\T_1 \equiv_{r+1} \T_2$ is 
equivalent to
$\T_{1,a} \equiv_r \T_{2,a}$ for every $a \in \S$, and 
$\T_1 \equiv_0 \T_2$ (the latter being abbreviated as 
($\ast$) below). 
Moreover, by Theorem~\ref{thm:edge-ambiguity-shortcut}, 
we have $\T_{1,a},\T_{2,a}\in\Cc_k(n,m,e)$. This enables
the following proof by induction on $r$:
\[
\begin{array}{lllll}
  \!\!
  \T_1     ~\equiv_{r+1}~ \T_2 & \quad \iff & \quad 
  \T_{1,a} ~\equiv_r~ \T_{2,a} ~~ (\forall a \in\S) ~~\text{ and } (\ast)
                             &             & \quad 
                                             \T_1     ~\equiv_N~
                                             \T_2. \\[1.25ex]
                             &             & \quad 
  \phantom{\T_{1,a}} ~\Updownarrow \text{(ind.~hyp.)} 
                             &             & \quad 
  \phantom{\T_1} ~\Updownarrow \\[1.25ex]
                             &            & \quad 
  \T_{1,a} ~\equiv_N~ \T_{2,a} ~~ (\forall a \in\S) ~~\text{ and } (\ast)
                             & \quad \iff & \quad
  \T_1     ~\equiv_{N+1}~ \T_2
\vspace{-6mm} 
\end{array}
\]
\qedhere
\end{proof}

Using Makanin's algorithm for solving word equations (and
even for deciding the existential theory of word equations, see
e.g.~\cite{die02makanin} for a modern presentation) we
obtain:

\begin{restatable}{proposition}{PropEffectiveTest}\label{prop:effective-test}
Given $n,m,e \in \Nat$, there is $N \in \Nat$ such that
$\S^{\le N}$ is a test set for every pair of  SST 
from $\Cc_k(n,m,e)$, and such an $N$ can be effectively computed.
\end{restatable}

\begin{proof}
  By Lemma~\ref{lem:test} we know that $N$ exists, and Makanin's
  algorithm allows to determine whether $\Ss_N,\Ss_{N+1}$ are
  equivalent, so to determine $N$ by Lemma~\ref{lem:effective}.
\end{proof}

We finally obtain the main result:

\begin{restate}{ThmMain}
\end{restate}

Of course, Theorem~\ref{thm:main} does not come with any 
complexity upper bound, mainly because of the Ehrenfeucht conjecture. 
The only known lower bound is $\PSPACE$-hardness, which holds even
for single-valued SST over unary output alphabets, and follows from a
simple reduction from universality of NFA.

Quite surprisingly, the exact complexity of equivalence is 
not known even for {\sl deterministic} SST, where the problem
is known to be between \NLOGSPACE{} and \PSPACE{}~\cite{AC11popl}. 
We also recall that equivalence of deterministic SST with 
{\sl unary output} can be checked in \PTIME{} using invariants~\cite{AlurDDRY13}.
Finally, we recall that the currently best upper bound for solving word equations is 
\PSPACE~\cite{plandowski04} (with even linear space requirement, 
as shown in~\cite{jez17icalp}).

\section{Conclusions}

Our paper answers to a question left open in~\cite{ad11}, showing
that the equivalence problem for finite-valued SST is decidable. We
followed a proof for one-way transducers due to Culik and
Karhumäki~\cite{ck86}, that is based on the Ehrenfeucht conjecture.
The main contribution of the paper is to provide the technical
development that allows to follow the proof scheme of~\cite{ck86}. We
believe that this development will also allow to obtain stronger
results. We conjecture that finite-valued SST can be
effectively decomposed into finite unions of unambiguous SST. This
would entail that in the finite-valued setting, 
two-way transducers and SST have the same expressive power, as is it
the case for single-valued transducers. If this holds with
elementary complexity, then the equivalence of single-valued SST (or
two-way transducers) could also be solved with elementary
complexity. We believe that the complexity is indeed elementary, and
leave this for future work.

\bibliography{bib,manual}

\end{document}